\documentclass[aps,prd,nofootinbib,notitlepage,unsortedaddress]{revtex4-1}
\usepackage{enumitem}
\usepackage{dsfont}
\usepackage{amsfonts}
\usepackage{amssymb}
\usepackage[latin1]{inputenc}
\usepackage[T1]{fontenc}
\usepackage{graphicx}
\usepackage{mathtools}
\usepackage{bm}
\usepackage{bbm}
\usepackage[normalem]{ulem}
\usepackage[mathscr]{euscript}
\usepackage[dvipsnames]{xcolor}

\usepackage[colorlinks=true]{hyperref}
\usepackage{amsthm}

\newcommand{\R}{\mathbb R}

\newcommand{\id}{\text{id}}

\newcommand{\D}{\mathrm{d}}

\newcommand{\db}{\bar\partial}

\newcommand{\y}{y}
\newcommand{\nab}{\mathring{\nabla}}

\def\br#1\er{\textcolor{red}{#1}}
\def\bb#1\eb{\textcolor{blue}{#1}}
\def\bo#1\eo{\textcolor{orange}{#1}}

 \newtheorem{prop}{Proposition}
\newtheorem{cor}[prop]{Corollary}
\newtheorem{lem}[prop]{Lemma}
\newtheorem{theor}[prop]{Theorem}
\newtheorem{rem}[prop]{Remark}
\newtheorem{defi}[prop]{Definition}
\newtheorem{ex}[prop]{Example}

\begin{document}
\title{Einstein-Kropina Metrics and Their Application in Finsler Gravity} %

%\FV{This is a good title, but crrently I'd propose "Einstein-Kropina metrics and their application in Finsler Gravity", or even "Einstein-Kropina metrics and their relativistic application"} 

\author{Sjors Heefer}
\email{s.j.heefer@tue.nl}
\affiliation{Department of Mathematics and Computer Science, Eindhoven University of Technology, Eindhoven, The Netherlands}
\author{Fidel F.\ Villase\~nor}
\email{fidel.fernandez@uclm.es}
\affiliation{Department of Mathematics, % Escuela T\'ecnica Superior de Ingenier\'ia Industrial de Ciudad Real, 
University of Castilla-La Mancha, Ciudad Real, Spain}
\author{Andrea Fuster}
\email{a.fuster@tue.nl}
\affiliation{Department  of Mathematics and Computer Science, Eindhoven University of Technology, Eindhoven, The Netherlands}

\begin{abstract}
% We generalize the Einstein condition for Kropina metrics obtained in the positive definite setting by Zhang and Shen to all signatures. As examples of Einstein-Kropina metrics, we construct new explicit ones in Riemannian signature and the first Lorentzian ones. Next, we use this Einstein condition to (locally) classify all Einstein-Kropina solutions to Pfeifer and Wohlfarth's field vacuum equation (including a possible cosmological constant term) for Finsler gravity in arbitrary dimension and in positive definite as well as Lorentzian signature. In dimension 4 or lower, no nontrivial solutions exist. In dimension 5 and higher, the solutions are precisely those for which the pseudo-Riemannian metric $\alpha$ is locally a product of the real line with a Ricci-flat pseudo-Riemannian metric of one dimension lower, and the 1-form $\beta$ is the unique (up to sign) unit 1-form on the first factor. In particular, and very surprisingly, all Einstein-Kropina solutions to the $\Lambda$-vacuum equation are Berwald and Ricci-flat, and the cosmological constant necessarily vanishes.

    We generalize the Einstein condition for Kropina metrics obtained in the positive definite setting by Zhang, Shen and others to all signatures. As examples of Einstein-Kropina metrics $L=L_{a,b}$, we construct new explicit positive definite ones and the first ones with Lorentzian signature. Next, we classify all Einstein-Kropina solutions to the vacuum equation of Finsler gravity by Pfeifer and Wohlfarth, in arbitrary dimension and including a possible cosmological constant  $\Lambda$. For the Lorentzian and positive definite cases, the local picture is as follows. In dimension 4 or lower, only the trivial solution exists: a Minkowski or Euclidean space, essentially.  In dimension 5 and higher, the solutions are those $L_{a,b}$ for which the metric $a$ is a product of the real line with a Ricci-flat metric (Lorentzian or Riemannian) and the vector field $b$ is the unique unit vector on the first factor. As a very surprising rigidity phenomenon, all Einstein-Kropina solutions to the $\Lambda$-vacuum equation are Berwald and Ricci-flat, and the cosmological constant necessarily vanishes.
\end{abstract}
\maketitle
\tableofcontents

% \section*{Structure paper}

% \begin{itemize}
%     \item Motivation
%     \begin{itemize}
%         \item Finsler metric that is both Einstein and weakly weakly Landsberg is automatically a $\Lambda$-vacuum sol to the field equation [remark later that this never happens for Kropina unless Ricci-flat]
%         \item No known examples of this combination. Examples of Einstein spaces? not much but in Kropina the Einstein condition is known. 
%     \end{itemize}
%     \item Generalize Einstein condition to arbitrary signatures. 
%     \begin{itemize}
%         \item Somewhere make a clear list with all identities for unit Killing etc.
%     \end{itemize}
%     \item  Give our known examples of nontrivial Einstein Kropina metrics
%     \begin{itemize}
%         \item Include 3D classification (const curv)
%         \item Examples in all dimensions except 4D (and in  1D trivially) \SH{what about 2D?}
%         \item Observe our properly Einstein examples not solutions. We kept finding this and then we were able to prove it never happens (see later).
%     \end{itemize}
%     \item $\Lambda$-Vacuum solutions characterization

% \end{itemize}

\section{Introduction}

 In this work, we study Kropina metrics of Einstein type in arbitrary signature. We cover aspects of intrinsic mathematical interest as well as those specific to the application of Kropina metrics in Pfeifer and Wohlfarth's Finslerian extension of general relativity: Finsler gravity \cite{Pfeifer:2011xi,Pfeiferetal2025}. \\

In contrast to pseudo-Riemannian Einstein metrics, their Finslerian generalizations---Einstein-Finsler metrics---have so far been studied primarily in the positive definite setting. Even there, the results are scarce:  characterization for special classes of metrics \cite{Robles_2003,ZhangXia2014}, some explicit examples \cite{BaoRobles,Cheng2017,Einstein_warped_Shen, TANG201883} and a Schur theorem in some cases \cite{villasenor:2024,YU2010290}.  In Lorentzian signature there are almost no results (we are only aware of \cite{TANG201883}), even though in this case, questions such as the existence of examples or the Schur theorem are of interest not only in geometry but also for relativistic applications, especially in cosmology. Now, in positive definite signature, Zhang, Shen and others have obtained a fairly rigid characterization of Einstein-Kropina metrics \cite{ZHANG201380,ZShen2014}. This makes the Kropina class a promising playground for a first investigation of Einstein-Finsler metrics in other signatures.  \\
 
 % with comparatively few results pertaining to other signatures. Moreover, while explicit examples of properly Einstein-Finsler metrics have been obtained in  , such examples are scarce, especially non-positive-definite ones. Some classification results exist, however, for special classes of Finsler metrics. In particular, in the positive definite setting, Zhang and Shen have obtained a characterization of Einstein-Kropina metrics in terms of two simple conditions for the defining pseudo-Riemannian metric $a$ and vector field $b$ \cite{ZHANG201380}. For this reason, the class of Kropina metrics provides a promising playground to start investigating the properties of Einstein-Finsler metrics also in other signatures. \\

%  Apart from their intrinsic mathematical interest, Einstein-Finsler metrics are also expected to play an important role in Finsler geometric extensions of general relativity, or Finsler gravity, for short. In Pfeifer and Wohlfarth's Finsler gravity framework \cite{Pfeifer:2011xi,Hohmann_2019}, spacetime is endowed with a Finsler metric $L$, which in $n>2$ spacetime dimensions %in the presence of a cosmological constant $\Lambda$ 
% satisfies Pfeifer and Wohlfarth's equation in vacuum (precise definitions will follow later):
% \begin{equation}
%    \left(n+2\right)\mathrm{Ric}-g^{ij}\mathrm{Ric}_{\cdot i\cdot j}\,L-2L\mathcal{P}=0.
%     \label{vacuum eq}
% \end{equation}
% \SH{check if we explain the $\cdot i$ notation for the vertical derivatives. I don't think we use it anywhere else.}

 Concerning the physical role of Einstein-Kropina (and more generally, Einstein-Finsler) metrics, in Pfeifer and Wohlfarth's Finsler gravity framework \cite{Pfeifer:2011xi,Hohmann_2019}, spacetime is endowed with a pseudo-Finsler metric $L$; in vacuum and $n>2$ spacetime dimensions, it satisfies the equation
\begin{equation}
   \left(n+2\right)\mathrm{Ric}-g^{ij}\bar\partial_i\bar\partial_j\mathrm{Ric}\,L-2\mathcal{P}\,L=0.
    \label{vacuum eq}
\end{equation}
(Precise definitions will follow later.)  In a forthcoming paper \cite{Lambda-vacuum-paper}, we generalize this equation to accommodate a possibly nonvanishing cosmological constant $\Lambda$. After adding the corresponding term in the natural Finsler generalization of the Einstein-Hilbert action, the variation of the action leads to
\begin{equation}
    \left(n+2\right)\mathrm{Ric}-g^{ij}\bar\partial_i\bar\partial_j\mathrm{Ric}\,L-2\mathcal{P}\,L+2\Lambda\,L=0.
    \label{lambda vacuum eq}
\end{equation}
This equation is to be interpreted as the Finslerian gravitational field equation in the presence of a cosmological constant $\Lambda$, in arbitrary spacetime dimension $n$: the \emph{$\Lambda$-vacuum equation}, in short. Interestingly, in contrast to the situation in general relativity, there are no known identities in Finsler geometry that force $\Lambda$ to be constant, and hence $\Lambda$ is in principle allowed to vary over the base manifold. For pseudo-Riemannian manifolds, the equation reduces to Einstein's vacuum equation with cosmological constant, and for $\Lambda=0$ to Pfeifer and Wohfarth's vacuum equation \eqref{vacuum eq}.\\

% We emphasize that even though one of the motivations of the Finsler gravity program has been to provide a geometric explanation for phenomena such as the accelerated expansion of the universe \textit{without} invoking a cosmological constant or dark matter (see \cite{Pfeiferetal2025} for promising recent developments), it is nevertheless natural to keep the assumptions of the theory as general as possible. In particular, there is no reason to exclude the possibility that $\Lambda\neq 0$ a priori. From this perspective, the $\Lambda$-vacuum equation \eqref{lambda vacuum eq} and its solutions are of clear importance.\\
One of the motivations for developing Finsler gravity is to provide a geometric explanation for phenomena such as the accelerated expansion of the universe without invoking a cosmological constant or dark matter (see \cite{Pfeiferetal2025} for promising recent developments). Nevertheless, in this work we keep the assumptions of the theory as general as possible since there is no reason to \emph{a priori} exclude the possibility that $\Lambda\neq 0$; this is also important for the sake of comparison with general relativity. From this perspective,  \eqref{lambda vacuum eq} and its solutions are of clear physical importance, in addition to being mathematically interesting.\\

It is well known that in general relativity, solutions to Einstein's vacuum field equations with cosmological constant are given precisely by (pseudo-Riemannian) Einstein spacetimes. While this correspondence does not identically generalize to the Finsler realm, it has the following close analog \cite{Lambda-vacuum-paper}. An Einstein-Finsler metric solves the $\Lambda$-vacuum equation \eqref{lambda vacuum eq} if and only if its Landsberg scalar $\mathcal P$ is isotropic, i.e. $\mathcal P = \mathcal P(x)$ (again, precise definitions follow later). Conversely, under some additional analyticity assumptions, solutions with isotropic Landsberg scalar must be Einstein. Clearly, then, Einstein metrics play an important role in the analysis of the field equation also in the Finslerian setting. \\

Motivated by these observations, we dedicate the present article to the study of Einstein-Kropina metrics, with two main goals: 
\begin{enumerate}[label=(\roman*)]
    \item to characterize Einstein-Kropina metrics in arbitrary signature and, especially, construct new, explicit examples; 
    \item to obtain and classify Einstein-Kropina solutions to the $\Lambda$-vacuum equation \eqref{lambda vacuum eq}.
\end{enumerate}
% First, regarding the Einstein condition for Kropina metrics, the situation is as expected and indeed identical to the positive definite one. A Kropina metric built from a pseudo-Riemannian one, $a$, and a vector field, $b$, written without loss of generality in such a way that $\langle b,b\rangle:=a(b,b)=1$, is Einstein if and only if $a$ is Einstein and $b$ is Killing (Th. \ref{theor:Krop_Einstein_cond}).  
% That means that to construct examples of Einstein-Kropina metrics, one must find pseudo-Riemannian Einstein metrics that admit a unit Killing vector field. We give explicit examples of such pairs $(a,b)$ in all nontrivial dimensions except %\footnote{(Remarkably, this is again one of those situations where the case $n=4$ turns out to behave very differently from the rest)}
% $n=4$ \textcolor{red}{with all possible signatures.} \SH{and signs of $\langle b,b\rangle$} These examples come in essentially two types: the pseudo-Riemannian metric $a$ is either a constant curvature space or an Einstein-Sasaki manifold, or a direct product of one of these with another Einstein metric. \\

First, regarding the Einstein condition for Kropina metrics, we prove that the situation is as in the positive definite case \cite{ZHANG201380}. Let $L=L_{a,b}$ be the Kropina metric built from a pseudo-Riemannian metric $a$, and a vector field $b$,  written without loss of generality in such a way that $\langle b,b\rangle:=a(b,b)=1$. Then, $L$ is Einstein if and only if $a$ is Einstein and $b$ is Killing (Theorem \ref{theor:Krop_Einstein_cond}). Einstein-Kropina metrics can thus be constructed by finding such pairs $(a,b)$. In all nontrivial dimensions except for $n=4$, we find proper (i.e. non-Ricci-flat) examples with all possible signatures. The examples (\S \ref{sebsec:examples}) are essentially of three types: the pseudo-Riemannian metric $a$ is of constant curvature, Einstein-Sasaki, or a direct product of these.\\

Interestingly, this leaves the case of dimension $4$ unsettled: there, it remains an open question whether proper Einstein-Kropina metrics exist at all. This problem, finding $4$D Einstein metrics $a$ that admit a unit Killing field $b$, seems to be much more difficult. Another noteworthy case is $n=3$ with Riemannian (resp. Lorentzian) signature, where the round metric on the sphere $S^3$ (resp. the $AdS_3$ metric) is the only possibility for $a$ such that the Kropina metric is Einstein. In higher dimensions, on the other hand, the picture is very different. For example, we find infinitely many distinct Einstein-Kropina metrics on $S^2\times S^3$.\\

Second, regarding the classification of solutions to the $\Lambda$-vacuum equation, we prove the following surprising result (Theorem \ref{general vacuum solutions}). A Kropina metric $L$ with building blocks $a$ and $b$ written (w.l.o.g.) such that $\langle b,b\rangle=1$ is an Einstein solution of \eqref{lambda vacuum eq} if and only if the following are satisfied:
\begin{itemize}
        \item $a$ is \emph{Ricci-flat},
        \item $b$ is Killing,
        \item $\nab_j b^k\nab_k b^i=0$, where $\nab$ is the Levi-Civita connection of $a$.
\end{itemize}
Moreover, in this case, $L$ is itself Ricci-flat and weakly weakly Landsberg ($\mathcal P = 0$), and necessarily $\Lambda = 0$. This is very unexpected: it establishes, in particular, that the Finslerian vacuum equation does not admit proper Einstein-Kropina solutions \emph{even when a cosmological constant is included}. \\

In case $a$ is Lorentzian or Riemannian, the result can be refined (Theorem \ref{riem and lor classification}). In these signatures, the condition $\nab_j b^k\nab_k b^i=0$ actually implies that $b$ is parallel, $\nab_j b^i = 0$. This, in turn, implies that $a$ and $b$ are locally of the form
\begin{align}
    % \label{eq:riem_lorentz_sols}
a=(\D x^1)^2+ \sum_{k,\ell=2}^n \bar{a}_{k\ell}(x^2,\ldots,x^n)\,\D x^k \otimes \D x^\ell, \qquad b=\partial_{x^1},
\end{align}
where $\bar{a}_{k\ell}$ are the components of a Ricci-flat $(n-1)$-dimensional pseudo-Riemannian metric. The converse is true as well: any $(a,b)$ locally of this form makes up an Einstein-Kropina metric that solves \eqref{lambda vacuum eq}, with $\Lambda=0$;  %In fact, it has been known for some time that any $(\alpha,\beta)$-metric constructed from \eqref{eq:riem_lorentz_sols} is a vacuum solution \cite{Heefer_2023_Finsler_grav_waves, Heefer:2024Thesis}.
the latter has been known for some time \cite{Heefer_2023_Finsler_grav_waves, Heefer:2024Thesis}. What is surprising is that, in the Einstein-Kropina class, these solutions exhaust the possibilities and this rigidity is unaffected by the addition of the cosmological constant. An extra result here is that at least when $b$ is timelike, Ricci-flatness of $L_{a,b}$, which was necessary for it to be a $\Lambda$-vacuum solution, becomes sufficient (Theorem \ref{theor:ricci-flat_sufficient_cond}).  \\

 In low dimensions, our conclusions become especially restrictive. In fact (Corollary \ref{cor:solutions_n=3,4}), for $n=3$ or $4$ and $a$ Riemannian or Lorentzian, the Einstein-Kropina metric $L_{a,b}$ is a ($\Lambda$-)vacuum solution if and only if $a$ is locally isometric to Euclidean space or Minkowski spacetime and $b$ is locally a generator of unit translations. Thus, nontrivial solutions exist only for $n\geq 5$. We end with an explicit example of a Lorentzian vacuum solution in the lowest nontrivial dimension $n=5$.

\section{Preliminaries}
\subsection{Conventions and basic definitions}\label{sec:conventions}
We start by outlining our conventions for the various geometric quantities that we use throughout the paper. For details, see e.g. \cite{Heefer:2024Thesis}.\\

Let $M$ be a connected smooth manifold with $n:=\dim M\geq 2$. If $x = (x^1,\dots,x^n)$ are coordinates on $M$, we denote by $(x^i,y^i)$ the induced coordinates on the tangent bundle $TM$. Most of the time, we will represent tensor fields by their components in these arbitrary charts, the Einstein convention being in force in the indices $i,j,k,\ell,m\in\left\{1,\ldots,n\right\}$. By abuse of notation, we will sometimes also denote a generic point in $TM$ by $(x,y)$. A \emph{conic subbundle} of $TM$ is an open subset $\mathcal A\subset TM$ which is conic in the sense that $(x,\lambda y)\in\mathcal A$ for any $(x,y)\in\mathcal A$ and $\lambda>0$, and which satisfies $\pi(\mathcal A) = M$, where $\pi:TM\to M$ is the canonical projection. %The latter property says that the fibers $\mathcal A_{x}$ of $\mathcal A$ are nonempty. 
To keep things as general as possible, we will work with the following minimal definition of Finsler structure, where we use the notation
\begin{align}
    \partial_i = \frac{\partial}{\partial x^i}, \qquad \bar\partial_i = \frac{\partial}{\partial y^i}.
\end{align}
\begin{defi}\label{def:Finsler_Lagrangian}
A pseudo-Finsler metric %\footnote{Sometimes this object is called \emph{pseudo-Finsler Lagrangian} and the word `metric' is reserved for a corresponding 1-homogeneous object. Here we simply refer to the 2-homogeneous object as the metric.} 
on $M$ is a smooth map $L:\mathcal A\to \R$ on a conic subbundle $\mathcal A$ of $TM$, such that
\begin{enumerate} [label=(\roman*)]
	\item $L$ is positively homogeneous of degree two:
	\begin{align}
	L(x,\mu y) =\mu^2 L(x, y)\,,\quad \forall \mu>0\,;
	\end{align}
	\item The \textit{fundamental tensor}, with components $g_{ij} := \db_i\db_j \left(\frac{1}{2}L\right)$, is nondegenerate.
\end{enumerate}
\end{defi}
%\SH{I think we probably want to keep it general and use this `minimal' definition in this paper.} 
\noindent Given a pseudo-Finsler metric, the \emph{Cartan tensor} and \emph{mean Cartan tensor} are defined as 
\begin{align}
    C_{ijk} := \tfrac{1}{2}\bar\partial_i g_{jk},\qquad C_{i} := g^{jk}C_{ijk}
\end{align}
(as usual, $\left(g^{ij}\right)$ denotes the inverse of $\left(g_{ij}\right)$). The \emph{spray coefficients} are defined as 
\begin{align}
    G^i := \tfrac{1}{2}g^{ik} \left(y^j \partial_j\bar\partial_k L - \partial_k L\right).
\end{align}
%\SH{(Note this is twice the $G^i$ in \cite{ZHANG201380}).}
The \emph{canonical nonlinear connection} can then be obtained as
\begin{align}
    N^j_i := \tfrac{1}{2}\bar\partial_i G^j.
\end{align}
In terms of the horizontal derivatives 
\begin{align}
    \delta_i:= \partial_i - N^j{}_i\db_j,
\end{align}
we may then define the nonlinear \emph{curvature tensor} and \emph{Ricci scalar} as follows: 

\begin{align}
    R(\partial_i,\partial_j) = -[\delta_i,\delta_j]  \eqqcolon R^k{}_{ij}\bar\partial_k, \quad  R^k_{ij} = \delta_i N^k_j - \delta_j N^k_i;
    \qquad \text{Ric} := R^i{}_{ij}y^j = y^j\left(\delta_i N^i_j - \delta_j N^i_i\right).
\end{align}
% \begin{align}
%     \text{Ric} = R^i{}_{ij}y^j = y^j(\delta_i N^i_j - \delta_j N^i_i) %= \partial_jG^j - \tfrac{1}{4}\bar\partial_iG^j\bar\partial_jG^i - \tfrac{1}{2}y^i\bar\partial_j\partial_iG^j + \tfrac{1}{2}G^i\bar\partial_j\bar\partial_iG^j.
% \end{align}
%\SH{This agrees with the convention in \cite{ZHANG201380}.} 

The \emph{Landsberg tensor} and \emph{mean Landsberg tensor} are defined as
\begin{align}
    P_{jk\ell} :=-\tfrac{1}{4}y_i \, \bar\partial_j\bar\partial_k\bar\partial_\ell G^i, \qquad P_j := g^{k\ell}P_{jk\ell}.
\end{align}
See e.g. \cite{Heefer:2024Thesis} for various equivalent definitions; in particular, the mean Landsberg tensor may be computed as $P_j = \nabla C_j$, where $\nabla$ is the dynamical covariant derivative. This may be defined as $\nabla \equiv y^j\nabla^B_{\delta_j}$ in terms of the Berwald connection $\nabla^B$ with horizontal Christoffel symbols ${}^B\Gamma^k_{ij} := \tfrac{1}{2}\bar\partial_i\bar\partial_jG^k$. We define the \emph{Landsberg scalar} as 
\begin{align}\label{eq:Landsb_scal_0hom}
    \mathcal P := g^{ij}\left(\bar\partial_j\left(\nabla P_i\right) + \nabla^B_{\delta_j}P_i \right).
\end{align}
This scalar combination appears in the Finslerian field equation \eqref{lambda vacuum eq}.

 \begin{defi}\label{def:Einstein}
   A pseudo-Finsler metric $L\colon\mathcal{A}\rightarrow\mathbb{R}$ is said to be: 
    \begin{itemize}
        \item \emph{Einstein} if there exists a function $\lambda\colon M\rightarrow\mathbb{R}$ (the \emph{Einstein coefficient}) such that $\mathrm{Ric}=\lambda\,L$;

        \item \emph{Weakly weakly Landsberg} if $\mathcal{P}=0$ on $\mathcal{A}$.
        
    \end{itemize}
\end{defi} 

\noindent In the Finslerian setting, being Einstein (e.g. \cite{Robles_2003} or \cite[Definition 2.6]{villasenor:2025}) means that the $0$-homogeneous version of the Ricci scalar is \emph{isotropic}, i.e., $\frac{\mathrm{Ric}}{L}(x,y)=\frac{\mathrm{Ric}}{L}(x)$($=\lambda(x)$, and at those $(x,y)\in \mathcal{A}$ with $L\neq 0$). We sometimes speak of an ``Einstein-Finsler'' metric $L$, but we emphasize that this may have any (definite or indefinite) signature, by which we mean that of $g_{ij}$. Among the pseudo-Finsler metrics of Einstein type are counted the \emph{Ricci-flat} ones, those with $\mathrm{Ric}=0$.

In the pseudo-Riemannian case, the Schur theorem guarantees that the Einstein coefficient must be constant (except in dimension 2); it is as of yet still an open question whether the same is true in the Finslerian setting. Only partial results have been established, see \cite{villasenor:2025} and the references therein.

Finally, being weakly weakly Landsberg is a particular case of having an isotropic Landsberg scalar ($\mathcal{P}(x,y)=\mathcal{P}(x)$). These two conditions are implied by the \emph{weak Landsberg} ($P_i=0$), the \emph{Landsberg} ($P_{ijk}=0$) and the \emph{Berwald} ($N^i_j(x,y)=\Gamma^i_{jk}(x)\,y^k$) ones.

\begin{rem}
    Suppose that $a$ is a pseudo-Riemannian metric on $M$ (so $a_x\colon T_xM\times T_xM\rightarrow\mathbb{R}$ is nondegenerate at each $x\in M$) and the pseudo-Finsler metric is that corresponding to $a$, namely given by $L(x,y)=a_x(y,y)$. Then, the nonlinear curvature and Ricci scalar are related to the Riemann curvature $\mathring R^k{}_{\ell ij}$ and Ricci tensor $\mathring R_{ij}$ of $a_{ij}$ via
\begin{align}
\label{eq:relation_curv_finsler_riemm}
R^k{}_{ij}(x,y)= \mathring  R^k{}_{\ell ij} (x)\,y^\ell, \qquad \mathrm{Ric}(x,y)= \mathring R_{ij}(x)\,y^i y^j.
\end{align}
Here we adhere to the following conventions:
\begin{align}\label{eq:affine_curvatures}
    \mathring R^k{}_{\ell ij} = \partial_i\mathring\Gamma^k_{j\ell} - \partial_j\mathring\Gamma^k_{i\ell} + \mathring\Gamma^k_{im}\mathring\Gamma^m_{j\ell} - \mathring\Gamma^k_{jm}\mathring\Gamma^m_{i\ell}, \qquad \mathring R_{\ell k} = \mathring R^i{}_\ell{}_{ik},
\end{align}
in terms of the Christoffel symbols $\mathring{\Gamma}^i_{jk}$ of $a$.

These conventions ensure that the Einstein coefficient of an Einstein pseudo-Riemannian metric agrees with the coefficient of its corresponding Einstein-Finsler metric, in the following sense. $L_a:=A$ is Einstein with Einstein coefficient $\kappa$ precisely if $\mathring R_{ij} = \kappa\,a_{ij}$, i.e. if $a$ is Einstein as a pseudo-Riemannian metric with Einstein constant $\kappa\in\mathbb{R}$. Moreover, if the signature of $a_{ij}$ is positive definite $(+\ldots +)$ or Lorentzian \emph{with mostly plus convention} $(+\ldots + -)$, then \eqref{eq:affine_curvatures} are such that the round sphere $S^n$ and de Sitter spacetime $dS_n$ have a positive Einstein constant and a positive sectional curvature. 

We stress that using a negative definite $(-\ldots -)$ or \emph{mostly minus} Lorentzian $(-\ldots -+)$ metric $a$ would give $S^n$ and $dS_n$ negative Einstein constant and sectional curvature. In some places of this paper, it will turn out to be convenient to allow both Lorentzian conventions, so it is important to keep this fact in mind.
\end{rem}

\subsection{Kropina metrics}

Let $a=a_{ij}(x)\,\D x^i\otimes\D x^j$ be a pseudo-Riemannian metric and $b=b^i(x)\,\partial_{x^i}$ be a vector field, always assumed to be smooth on $M$. We shall work with the functions\footnote{In the literature on positive definite Finsler geometry, one usually prefers the quantities $\alpha,\beta$, where $\alpha  = \sqrt{A}$. To avoid ill-defined square roots, we always work directly with $A$, partly keeping the traditional notation for $(\alpha,\beta)$-metrics and partly following that of \cite{Voicu_spacetime_cond_2023}. }%, which is then well defined as $a_x(y,y)$ is guaranteed to be nonnegative, in contrast to the case at hand here. \bb We partly follow the notation of \cite{Voicu_spacetime_cond_2023}. \eb}
$A,\beta\colon TM\rightarrow\R$ given, for $y\in T_xM\subset TM$, by
\begin{equation}
    A(x,y):=a_x(y,y),\qquad\beta(x,\y):=a_x(b_x,y).
    \label{A and beta}
\end{equation}
As $\beta$ is linear in $\y$, it can be viewed as a $1$-form on $M$: the $a$-dual of $b$ (but it is this vector field that we consider as primary). Another important function shall be $\langle b,b\rangle = a(b,b)$, given by
\begin{align}
    \langle b,b\rangle(x):=a_x(b_x,b_x). 
\end{align}
When $a$ is positive definite, $\langle b,b\rangle$ is the squared norm of the vector field $b$, but in general, it may have any sign, consistent with the different possibilities for the `causal character' of $b$.

\begin{defi}\label{def:Kropina_metric}
    A \emph{Kropina pre-metric} is a function $L\colon\mathcal{A}\rightarrow\R$ of the form $L = A^2/\beta^2$, where $\mathcal A$ is a conic subbundle of $TM$. If, additionally, $L$ satisfies the axioms of a pseudo-Finsler metric, then $L$ is called a \emph{Kropina metric}.
\end{defi}

%\begin{defi}\label{def:Kropina_metric}
%    A \emph{Kropina metric} is a pseudo-Finsler metric $L\colon\mathcal{A}\rightarrow\R$ of the form $L = A^2/\beta^2$.
    %\[
    %L(x,\y)=\frac{A(x,\y)^2}{\beta(x,\y)^2}\qquad \forall(x,\y)\in\mathcal{A}.
    %\]
%\end{defi}

\noindent Clearly, any Kropina metric is the restriction to $\mathcal{A}$ of the Kropina pre-metric $L^{\rm max}\colon\mathcal{A}^{\rm max}\rightarrow \R$ defined on the maximal domain 
\begin{align}
    \mathcal{A}^{\rm max}:=\left\{(x,\y)\in TM\colon \beta(x,\y)\neq 0\right\},\qquad L^{\rm max}:=\frac{A^2}{\beta^2}.
\end{align}
The pair $(\mathcal{A}^{\rm max},L^{\rm max})$ is unique given $(a,b)$. Since the domains of our Kropina metrics $L$ must satisfy $\pi(\mathcal{A})=M$ and $\mathcal{A}\subseteq\mathcal{A}^{\rm max}$, it already follows that $b$ cannot have zeros on $M$.

For some purposes, it is more convenient to work with $L^{\rm max}$ rather than $L$. In particular, we can relate the signature of $L^{\rm max}$ with that of $A$, using a technique from \cite[Lemma  1.1.2]{ChernShen2005}. We consider an interpolation between the two governed by a parameter $t\in\left[0,1\right]$. %Define
% \[
% \Psi_t\colon\R\setminus\left\{0\right\}\longrightarrow\R,\qquad\Psi_t _t(s):=1-t+ts^{-1};
% \]
% \[
% L^{t}\colon\mathcal{A}^{\rm max}\cap\left\{A(x,\y)\neq 0\right\}\longrightarrow\R,\qquad L^{t}(x,\y):=A(x,\y) \Psi_t _t(\frac{\beta(x,\y)^2}{A(x,y)})=\left(1-t\right)A(x,\y)+t\frac{A(x,\y)^2}{\beta(x,\y)^2}.
% \]
% ---------------------\\
For each $t\in[0,1]$, define
\begin{align}
    L^{t}\colon\mathcal{A}^{\rm max}\longrightarrow\R,\qquad L^{t}(x,\y):=\left(1-t\right)A(x,\y)+t\frac{A(x,\y)^2}{\beta(x,\y)^2},
\end{align}
and denote the fundamental tensor of $L^t$ by $g^t_{ij} = \tfrac{1}{2}\bar\partial_i\bar\partial_j L^t$. Note that whenever $A\neq 0$, we can write $L^t$ as
\begin{align}
    L^t = A\Psi_t(s), \qquad s = \beta^2/A, \qquad \Psi_t(s):=1-t+ts^{-1}.
\end{align}
The result \cite[Eq. (A1)]{Voicu_spacetime_cond_2023} states that for such an expression,  one has 
\begin{align}\label{eq:det_Voicu}
    \det (g^t_{ij})= \Psi_t ^2(\Psi_t -s\Psi_t ')^{n-3} \frac{\partial}{\partial s}\left[\left(s-\langle b,b\rangle\right)\frac{(\Psi_t -s\Psi_t ')^2}{\Psi_t }\right]\det a_{ij}
\end{align}
whenever $\Psi_t -s\Psi_t '\neq 0$ (the prime $'$ means derivative with respect to $s$). This leads to the following explicit formula.
\begin{lem} The determinant of $g^t_{ij}$ is given by 
    \begin{align}
        \det (g^t_{ij})= \left(1-t+\frac{2t}{s}\right)^{n-2}\left[ (1-t)^2 + \frac{3\langle b,b\rangle(1-t)t}{s^2} + \frac{2\langle b,b\rangle t^2}{s^3} \right]\det a_{ij},
    \end{align}
\end{lem}
whenever $1-t+\frac{2t}{s}\neq 0$, with the understanding that $1/s=0$ if $A=0$.
\begin{proof}
This is mostly a matter of computation. Suppose first that $A\neq 0$, so that \eqref{eq:det_Voicu} applies. We have $\Psi_t -s\Psi_t ^\prime=1-t+2ts^{-1}$ and hence 
\begin{align}
    \frac{\left(\Psi_t -s\Psi_t ^\prime\right)^2}{\Psi_t }%=\frac{\left(1-t+2ts^{-1}\right)^2}{1-t+ts^{-1}}
=\frac{\left[\left(1-t\right)s+2t\right]^2}{\left(1-t\right)s^2+ts},
\end{align}
with derivative
\begin{align}
    \frac{\partial}{\partial s}\left[\frac{\left(\Psi_t -s\Psi_t ^\prime\right)^2}{\Psi_t }\right]=\frac{2\left[\left(1-t\right)s+2t\right]\left(1-t\right)\left[\left(1-t\right)s^2+ts\right]-\left[\left(1-t\right)s+2t\right]^2\left[2\left(1-t\right)s+t\right]}{\left[\left(1-t\right)s^2+ts\right]^2}.
\end{align}
According to \eqref{eq:det_Voicu}, $\det g^t_{ij}/\det a_{ij} = (\Psi_t -s\Psi_t ^\prime)^{n-2}\psi_t = (1-t+2ts^{-1})^{n-2}\psi_t$, where $\psi_t$ is given by
\begin{align}
    \begin{split}
    \psi_t = &\quad\frac{\Psi_t ^2}{\Psi_t -s\Psi_t ^\prime}\frac{\partial}{\partial s}\left[\left(s-\langle b,b\rangle\right)\frac{\left(\Psi_t -s\Psi_t ^\prime\right)^2}{\Psi_t }\right] =\Psi_t \left(\Psi_t -s\Psi_t ^\prime\right)+\left(s-\langle b,b\rangle\right)\frac{\Psi_t ^2}{\Psi_t -s\Psi_t ^\prime}\frac{\partial}{\partial s}\left[\frac{\left(\Psi_t -s\Psi_t ^\prime\right)^2}{\Psi_t }\right] \\
    &=\left(1-t+ts^{-1}\right)\left(1-t+2ts^{-1}\right)+\left(s-\langle b,b\rangle\right)\frac{2s\left(1-t\right)\left[\left(1-t\right)s^2+ts\right]-s\left[\left(1-t\right)s+2t\right]\left[2\left(1-t\right)s+t\right]}{s^4} \\
    &=\left(1-t\right)^2+\frac{3\left(1-t\right)t}{s}+\frac{2t^2}{s^2}+\left(\langle b,b\rangle-s\right)\left(\frac{3(1-t)t}{s^2} + \frac{2t^2}{s^3}\right) \\
    &= \left(1-t\right)^2+\langle b,b\rangle\left(\frac{3(1-t)t}{s^2} + \frac{2t^2}{s^3}\right).
\end{split}
\end{align}
This yields the desired formula whenever $A\neq 0$. Then, by continuous extension, the formula must hold on all of the maximal domain $\mathcal A^{\text{max}}$.
\end{proof}
\noindent As a special case, by setting $t=1$ in the Lemma, we obtain the determinant of the fundamental tensor of a Kropina (pre-)metric. %The result shows that $\langle b,b\rangle$ is necessarily nowhere vanishing and hence, that we may always assume w.l.o.g. that $\langle b,b\rangle=1$, by rescaling $a$, as discussed above. BELOW
\begin{cor}\label{cor:det-krop}
    The determinant of the fundamental tensor components of a Kropina pre-metric is 
    \begin{align}
        \det g_{ij} = 2^{n-1}\langle b,b\rangle\left(\frac{A}{\beta^2}\right)^{n+1}\det a_{ij}.
    \end{align}
    Thus, $g$ is nondegenerate at a given $(x,y)\in \mathcal A^\text{max}$ if and only if $A(x,y)\neq 0$ and $a_x(b_x,b_x)\neq 0$.
    
    \noindent In particular, if the pair $(a,b)$ constitutes a Kropina metric, then the vector field $b$ is nowhere null with respect to $a$.
\end{cor}

\begin{rem}
    Before we proceed to determine the signature of a Kropina (pre-)metric, it is useful to observe that $a$ and $b$ are not unique given $(\mathcal{A}^{\rm max},L^{\rm max})$. Indeed, $L^{\rm max}$ remains the same under changes $(a,b)\mapsto (a,-b)$ and $(a,b)\mapsto (fa,b)$ whenever $f\in C^\infty(M)$ is nowhere vanishing (recall \eqref{A and beta}). This has important implications:
    \begin{enumerate}
        \item If $b$ is taken as a fixed datum, $L^{\rm max}$ becomes an invariant of the conformal class of $a$; in indefinite signature, this means that it only depends on the bundle of null cones associated with $a$. In the Lorentzian case, this bundle is precisely  the  (unoriented) causal structure of $(M,a)$.

        \item What is more, the invariance under $(a,b)\mapsto (-a,b)$ makes the choice of convention for Lorentzian metrics irrelevant. So, one may take the Lorentzian signature of $a$ to be $(- +\ldots +)$ or $(+ - \ldots -)$ indistictly.

        \item Assuming that $(a,b)$ constitutes a Kropina metric, we may set $f = \frac{1}{a(b,b)}\in C^\infty(M)$ above (Corollary \ref{cor:det-krop}) to obtain a new pseudo-Riemannian metric $\widetilde a = fa$ satisfying
        \begin{align}
            L^{\rm max} =\frac{\widetilde A^2}{\widetilde\beta^2},\qquad \widetilde{a}(b,b)=1.
        \end{align}
        This implies that any Kropina metric can be written in the form $L = A^2/\beta^2$ with $\langle b,b\rangle= 1$.
    \end{enumerate}
    We will repeatedly use these facts, particularly the last one, throughout the remainder of the paper.
    
%     Indeed, recalling \eqref{A and beta}:
% \begin{enumerate}
%     \item $L^{\rm max}$ is an invariant of the conformal class of $a$, in the sense that it stays the same under changes $(a,b)\mapsto(e^{2u}a,b)$ for arbitrary $u\in C^\infty(M)$. For $a$ of indefinite signature, this means that, once $b$ is fixed, $L^{\rm max}$  depends only on the bundle of null cones associated with $a$.  In the Lorentzian case, this is the (unoriented) causal structure of $a$.

%     \item $L^{\rm max}$ does not change under $(a,b)\mapsto(-a,b)$ either. For the Lorentzian case, this means that it is irrelevant whether one takes $a$ to be of signature $(-,+,\ldots,+)$ or $(+,-,\ldots,-)$.
% \end{enumerate}
% All in all, $L^{\rm max}$ remains the same under $(a,b)\mapsto(fa,b)$ provided that the function $f=f(x)$ is nowhere vanishing. If, moreover, $\langle b,b\rangle=a(b,b)$ is nowhere vanishing (i.e., the vector field $b$ is nowhere null), we may set $f = 1/a(b,b)$ to obtain a new pseudo-Riemannian metric $\widetilde a = fa = \frac{a}{a(b,b)}$ satisfying
% \begin{align}
%     L_\text{max} = \widetilde A^2/ \widetilde\beta^2, \qquad \widetilde{a}(b,b)=1.
% \end{align}
% In particular, it follows from Cor. \ref{cor:det-krop} that if the pair $(a,b)$ constitutes a Kropina metric, then $b$ is necessarily nowhere null. This implies that any Kropina metric can be written in the form $L = A^2/\beta^2$ with $\langle b,b\rangle= 1$. Throughout the remainder of the paper, we will use this fact repeatedly.
\end{rem}

The following proposition relates the signature of $L$ to that of $a$. Here we will assume that $\langle b,b\rangle=a(b,b)>0$. If, on the other hand, $\langle b,b\rangle<0$, an analogous result can be easily obtained by performing a change $a\mapsto-a$, which leaves $L$ unchanged but changes the sign of $\langle b,b\rangle$.

\begin{prop}\label{prop:signature_Kropina_metric} 
Let $L=A^2/\beta^2$ be a Kropina pre-metric with fundamental tensor $g_{ij}$ and assume that $\langle b,b\rangle>0$ everywhere. Then for all $y\in T_xM$ with $a_x(y,y) = A(x,y) > 0$, the signature of $g_{ij}(x,y)$ is the same as the signature of $a_{ij}$. In particular:
\begin{itemize}
    \item if  $a$ is positive definite, then $L$ is positive definite on all of $\mathcal A^\text{max}$;
    \item  if $a$ is Lorentzian and $b$ is timelike (spacelike), then $L$ has Lorentzian signature within the timelike cones (resp., spacelike sector) of $a$.
\end{itemize}
\end{prop}
\begin{proof}
If we assume w.l.o.g. that $\langle b,b\rangle=1$, then the determinant of the interpolating metric simplifies to
\begin{align}
    \det g^t_{ij}= (1-t+2ts^{-1})^{n-2}\left[ (1-t)^2 + \frac{3(1-t)t}{s^2} + \frac{2t^2}{s^3} \right]\det a_{ij}.
\end{align}
Now fix some $y\in T_xM$ with $a_x(y,y) = A(x,y) > 0$, which is equivalent to $s>0$. Evaluating the determinant at this vector $y$, the term in round brackets is a convex combination of two positive terms, hence positive; and the term in square brackets is clearly also positive. Therefore the whole prefactor multiplying $\det a_{ij}$ is positive for all $t\in [0,1]$, which implies that the signature at $t=0$ must be the same as the signature at $t=1$. In other words, the signature of $g_{ij}(x,y)$ is the same as that of $a_{ij}$.
\end{proof}
\color{black}

\noindent We end this section with some useful formulas for Kropina metrics. It is straightforward to derive these using computer algebra software such as the xAct package for Wolfram Mathematica \cite{MartinGarcia:xAct}. The fundamental tensor of a Kropina metric $L = A^2/\beta^2$ is given by
\begin{align}\label{eq:fund_tensor_Krop}
    g_{ij} = \frac{2A}{\beta^2}a_{ij} + \frac{3A^2}{\beta^4}b_ib_j - \frac{4A}{\beta^3}(b_iy_j+ y_ib_j) +\frac{4}{\beta^2}y_iy_j,
\end{align}
with inverse
\begin{align}
    g^{ij} = \frac{\beta^2}{2A}a^{ij} - \frac{\beta^2}{2A\,\langle b,b\rangle}b^ib^j + \frac{\beta^3}{A^2\,\langle b,b\rangle}(b^iy^j+ y^ib^j) +\frac{\beta^2}{A^3}\left(A-\frac{2\beta^2}{\langle b,b\rangle}\right)y^iy^j.
\end{align}
The spray of a Kropina metric $L = A^2/\beta^2$ is given by
\begin{align}\label{eq:spray_Krop}
   G^i = \mathring{\Gamma}^i_{jk}y^jy^k - \frac{A}{\beta}s^i{}_jy^j + \frac{1}{\langle b,b\rangle}\left(\frac{A}{\beta}s_{jk}b^jy^k + r_{jk}y^j y^k\right)b^i - \frac{2}{\langle b,b\rangle}\left(s_{jk}b^jy^k + \frac{\beta}{A}r_{jk}y^j y^k\right)y^i,
\end{align}
where
\begin{align}
    r_{ij} = \tfrac{1}{2}(\mathring\nabla_jb_i + \mathring\nabla_ib_j), \qquad s_{ij} = \tfrac{1}{2}(\mathring\nabla_jb_i - \mathring\nabla_ib_j).
\end{align}
\begin{rem} \label{rem:raising_and_lowering}
    Starting in \eqref{eq:spray_Krop} and for the rest of the paper, indices are raised with $\left(a^{ij}\right)$ and lowered with $\left(a_{ij}\right)$ (to be distinguished from  operations performed with $g$ and its inverse). Also, quantities with a circle above them relate to $a$, e.g. $\mathring\nabla$ is the Levi-Civita connection of $a$. % and $\mathring{\Gamma}^i_{jk}$ are the corresponding Christoffel symbols.
\end{rem}

\section{Einstein-Kropina metrics}

\subsection{The Einstein condition for Kropina metrics}\label{sec:Einstein condition}

%A pseudo-Riemannian metric $a$ is said to be Einstein with Einstein coefficient $\kappa$ if $\mathring R_{ij} = \kappa a_{ij}$ for some $\kappa:M\to\R$. From \eqref{eq:relation_curv_finsler_riemm} it is clear that this is equivalent to $\mathring{\text{Ric}} = \kappa \mathring L$ if we identify $a$ with the pseudo-Finsler metric $\mathring L = A$. More generally, a pseudo-Finsler metric $L$ is said to be of Einstein type with Einstein coefficient $\lambda$ if $\text{Ric} = \lambda L$ for some $\lambda:M \to \R$. In what follows, we will consistently denote the Einstein coefficient of a pseudo-Riemannian metric by $\kappa$ while reserving the symbol $\lambda$ for the Einstein coefficient of the corresponding Kropina metric.

% In the pseudo-Riemannian case, the Schur theorem guarantees that the Einstein coefficient must be constant (except in dimension 2), it is as of yet still an open question whether the same is true in the Finslerian setting; only partial results have been established, see e.g. \cite{villasenor:2025}. 

In the positive definite setting, Zhang and Shen \cite{ZHANG201380} have shown that a Kropina metric of dimension $\geq 2$ is of Einstein type if and only if $a$ and $b$ can be chosen such that $a$ is Einstein and $b$ is a Killing vector of $a$ of unit norm. Here we report that their proof generalizes essentially verbatim to arbitrary signatures in the case $n\geq 3$, and with only a slight modification in dimension $n=2$. Since the proof is so similar, we omit it here and include only a sketch of the argument in Appendix \ref{app:proof_Einstein_cond}.

\begin{theor}[Einstein condition for Kropina metrics]\label{theor:Krop_Einstein_cond}
    A Kropina metric $L=\left(\frac{A}{\beta}\right)^2$ (of arbitrary signature), with $n\geq 2$ and written such that $\langle b,b\rangle=1$, is Einstein if and only if $a$ is Einstein and $b$ is Killing w.r.t. $a$. In this case, if $\mathring R_{ij} = \kappa \,a_{ij}$ then $\mathrm{Ric} = \tfrac{\kappa}{4}L$.
\end{theor}

\noindent One of the main purposes of this paper is to give new examples of non-Ricci-flat Einstein-Kropina metrics. We can, however, immediately rule out their existence in 2D: the following result shows that, apart from global topology, there are essentially only two 2D Einstein-Kropina metrics, and they are Ricci-flat.

\begin{cor}\label{cor:unitKV_2D}
    Let $L = A^2/\beta^2$ be an Einstein-Kropina metric written such that $\langle b,b\rangle=1$ and assume that $\dim M=2$. Then $a$ is flat, $b$ is parallel and w.r.t. $a$, and $L$ can be written in local coordinates as
    \begin{align}
        L(x^1,x^2,y^1,y^2) = \frac{\left[(y^1)^2\pm (y^2)^2\right]^2}{(y^1)^2}.
        \label{eq:2D_Einstein_Kropina}
    \end{align}
\end{cor}
\begin{proof}
    %It suffices to show that any unit Killing vector field in 2D is parallel, as it is well known that any $(\alpha,\beta)$-metric with parallel $\beta$ is Berwald, and a Berwald space that is also Einstein is either Ricci-flat or pseudo-Riemannian. So %
    Suppose $a$ is a 2D pseudo-Riemannian metric with a unit Killing vector field, $\langle b,b\rangle=1$. Then, locally, $a$ and $b$ can be written in adapted coordinates $(u,v)$ as
    \begin{align}\label{eq:adapted_2D_coords}
        a = \D u^2 + 2f(v)\,\D u\,\D v + g(v)\,\D v^2 = \left(\D u + f(v)\,\D v\right)^2 + h(v)\,\D v^2, %= \beta^2 + h(v)\,\D v^2, 
        \qquad b = \partial_u
    \end{align}
    for suitable functions $f,g,h$. Introducing new coordinates $x^1,x^2$ defined by $\D x^1 = \D u + f(v)\,\D v$, $\D x^2 = \sqrt{|h(v)|}\,\D v$, the metric locally attains the form
    \begin{align}
        a = \left(\D x^1\right)^2 \pm \left(\D x^2\right)^2.
    \end{align}
    Furthermore, recalling that $\beta$ is the dual $1$-form of $b$, one sees in \eqref{eq:adapted_2D_coords} that 
    \begin{align}
        \beta=\D u + f(v)\,\D v=\D x^1, 
    \end{align}
     leading to the desired expression for $L$.
\end{proof}

\begin{rem}
    The (local) % situation
    landscape of Einstein-Kropina metrics in $n=2$ is completely described by Corollary  \ref{cor:unitKV_2D}; since \eqref{eq:2D_Einstein_Kropina} is flat and Berwald, we can already determine when they solve the $\Lambda$-vacuum equation \eqref{lambda vacuum eq}.

    Indeed, a 2D Einstein-Kropina metric (necessarily \eqref{eq:2D_Einstein_Kropina}) always solves \eqref{lambda vacuum eq} if $\Lambda=0$, and never if $\Lambda\neq 0$. Notice that this  parallels the general relativity situation: a 2D pseudo-Riemannian metric (necessarily Einstein) always solves the classical Einstein equation if $\Lambda=0$, and never if $\Lambda\neq 0$. 
    
    Since we have settled the study of Einstein-Kropina metrics in dimension $2$, we will assume that $n\geq 3$ in the remainder of the article whenever convenient.

    % \bb This is also the reason for having introduced a factor of $n-2$ in the $\Lambda$-vacuum equation \eqref{lambda vacuum eq}. In this way, any 2D Einstein-Kropina metric solves \eqref{lambda vacuum eq}, much like the way any 2D pseudo-Riemannian metric automatically solves the classical Einstein equation. \eb \SH{Hmm I think I see what you mean. But is this really so similar as this sentence suggests? In our case any 2D L will be a solution for any cosmological constant due to these conventions, no? Whereas in the GR case, any metric will be a solution but only for $\Lambda=0$. Maybe its a bit misleading to state it like this?}
\end{rem}

\noindent In dimension $3$, the situation is also very constrained, since any 3D Einstein pseudo-Riemannian metric is a constant curvature one, and only a few of these admit unit Killing vector fields:

\begin{lem}\label{lem:unitKV_const_curv}
    Let a Riemannian or Lorentzian constant curvature space admit a unit Killing vector (one with $\langle b,b\rangle=1$). Then, it is either flat or locally isometric to:
\begin{itemize}
    \item an odd-dimensional round sphere $S^n$ with $n\geq 3$;
    \item an odd-dimensional anti-de Sitter spacetime $AdS_n$ in mostly minus convention, with $n\geq 3$.
\end{itemize}

\end{lem}
\noindent This fact can be put together from several different results that are well known in the literature; however, we have included a sketch of a self-contained proof in Appendix \ref{app:unitKV_const_curv}. The following result is then immediate from Theorem \ref{theor:Krop_Einstein_cond}.

\begin{cor}
    Let $L = A^2/\beta^2$ be an Einstein-Kropina metric written such that $\langle b,b\rangle=1$ and assume that $\dim M=3$. Then $a$ is locally isometric to $S^3$, $AdS_3$ or a flat $\R^3$.
\end{cor}
\noindent In dimension $4$, the situation turns out to be particularly nontrivial: we have not been able to construct 4D non-Ricci flat Einstein-Kropina metrics. However, understanding which 4D Einstein metrics admit unit Killing fields is work in progress, in sight of the  open question whether nontrivial examples exist. %and in fact, we have obtained some partial results suggesting that such examples might not exist at all. 

Finally, the examples below show that in any dimension $5$ or greater, Einstein-Kropina metrics do exist, in Riemannian as well as Lorentzian signature.

\subsection{Examples of Einstein-Kropina metrics} \label{sebsec:examples}

% \SH{Personally, I would go in the following order: 
% \begin{enumerate}
%     \item  constant curvature spaces in all (possible) dimensions (+ noting that in $n=3$ these exhaust the list of possibilities)
%     \item products of constant curvature spaces: general statement that we can make examples in any dimension $n>3$ in both Riemannian/Lorentzian signature (for $a$ at least\dotes) + maybe give a list or nice table of all possibilities in $n=5$ (I would keep this relatively short)
%     \item Einstein-Finsler-Sasaki metrics
% \end{enumerate}
% }

\subsubsection{Constant curvature spaces}

Einstein-Kropina metrics with $a$ of constant curvature are the simplest possible ones, and Lemma \ref{lem:unitKV_const_curv} gives limited options for them. For proper examples, the metric $a$ must be that of an odd-dimensional $S^n$ or an odd-dimensional $AdS_n$. In both of these cases, it is easy to find unit Killing vectors by viewing them as embedded in a flat space of dimension $n+1$. \\

\begin{ex}[The odd-dimensional round sphere $S^n$]
    Let $n=2k+1$ and consider the embedded
    \begin{align}
        S^n = \big\{ x\in\R^{n+1}\,:\, \sum_{i=1}^{n+1} (x^i)^2=1\big\} \label{n-sphere}
    \end{align}
    with its induced Riemannian metric $a_{S^n}$. This metric is Einstein with Einstein constant $\kappa = n-1$. The following is a unit Killing vector field,
    \begin{align}
        b_{S^n} := \sum_{j=1}^{(n+1)/2}\left(x^{2j-1}\partial_{x^{2j}}-x^{2j}\partial_{x^{2j-1}} \right), \label{KV n-sphere}
    \end{align}
    generating the standard isometric $U(1)$-action on $S^n$ given by $k+1$ %pairs of 
    identical rotations in $k+1$ orthogonal $2$-planes. Then the positive definite Kropina metric $L_{S^n} := A_{S^n}^2/\beta_{S^n}^2$ is Einstein with Einstein constant $\lambda = (n-1)/4$ (as per Theorem \ref{theor:Krop_Einstein_cond}). A similar example was given in \cite{ZHANG201380} in the $n=3$ case.  
\end{ex} 

\begin{ex}[Odd-dimensional $AdS_n$] 
\label{ex:ads} 
    Let $n = 2k+1$ and consider
    \begin{align}
        \text{AdS}_n = \big\{ (x^1,\dots,x^{n-1},u,v) \in\R^{n-1,2}\,:\,\sum_{i=1}^{n-1} (x^i)^2 - u^2 - v^2=-1\big\}, \label{AdS_n}
    \end{align}
    with metric $a_{_{\text{AdS}_n}}$ induced by the flat one $d s^2 = \sum_{i=1}^{n-1} \left(\D x^i\right)^2 - \D u^2 - \D v^2$ on $\R^{n-1,2}$. It has Einstein constant $\kappa = -(n-1)$. The following is a Killing vector field,
\begin{align}
    b_{_{\text{AdS}_n}} := \sum_{j=1}^{(n-1)/2}\left(x^{2j-1}\partial_{x^{2j}}-x^{2j}\partial_{x^{2j-1}} \right) + u\,\partial_{v}-v\,\partial_{u}, \label{KV AdS_n} %\\
    %
    %= (x^2\partial_1 - x^1\partial_2) + (x^4\partial_3 - x^3\partial_4) +\dots + (x^{n-1}\partial_{n-3} - x^{n-3}\partial_{n-1}) + (s\partial_t - t\partial_s)
\end{align}
again generating an isometric $U(1)$-action given by individually rotating $k+1$ pairs of coordinates in $k+1$ orthogonal 2-planes.
With these conventions, the signature of $\text{AdS}_n$ is $(+ \ldots + -)$ and $a_{_{\text{AdS}_n}}(b_{_{\text{AdS}_n}},b_{_{\text{AdS}_n}})=-1$, so $b_{_{\text{AdS}_n}}$ is timelike. This means that the Kropina metric $L_{_{\text{AdS}_n}} := A_{_{\text{AdS}_n}}^2/\beta_{_{\text{AdS}_n}}^2$ is an Einstein-Finsler metric, with Einstein constant\footnote{\label{footnote:sign_flip}Note that the chosen signature implies $\langle b,b\rangle=-1$, rather than $\langle b,b\rangle=1$ as we have assumed throughout the above. To apply %results like Theorem \ref{theor:Krop_Einstein_cond},
our previous results, one technically first needs to change $a\mapsto -a$ (which leaves the resulting Kropina metric unchanged), leading to $\langle b,b\rangle=1$. However, this leads to $(-\ldots - +)$ signature and, together with our other conventions, a positive curvature for $AdS_n$. This may be confusing, as $AdS_n$ is usually understood to have a negative curvature; we therefore stick to $(+\ldots+-)$ signature and $\langle b,b\rangle=-1$ in the first part of this example. At any rate, this explains the overall sign difference between the Einstein constants $\kappa$ and $\lambda$. %The overall sign difference with respect to the Einstein constant of a is due to the signature flip $a\mapsto -a$, needed to apply our results %(see footnote \ref{footnote:signature_flip}).
} $\lambda=(n-1)/4$, that has signature $(-\dots -+)$ in the timelike cones $a_{_{\text{AdS}_n}}(y,y)<0$ (in light of Proposition \ref{prop:signature_Kropina_metric}). To the best of our knowledge, this is the first example of an Einstein-Kropina metric with indefinite, Lorentzian signature. 
\end{ex}

\subsubsection{Product spaces}
Further novel explicit examples can be obtained by considering direct products of pseudo-Riemannian metrics; notably, these exist in all dimensions $n\geq 5$. If both factors are Einstein with the same non-zero Einstein constant, then the product will also be Einstein \cite[Proposition 1.99]{besse1987einstein}. Similarly, if one factor has a unit Killing field, then (the natural lift of) this vector field will still be unit and Killing for the product metric. Note that in $n=5$, both factors \emph{must} be of constant curvature in order to be non-Ricci-flat Einstein. (This restriction does not hold anymore in $n\geq 6$, as Einstein does not imply constant curvature for factors of dim. greater than $3$, see e.g. \ref{Einstein-Sasaki}.) %Since the Einstein \bb constants \eb of the two factors must match, 
For simplicity, we focus here on the case where both factors are of constant curvature. Hence, we consider one of the factors to be an odd-dimensional sphere or anti-de Sitter spacetime. The possible combinations are listed schematically in table \ref{tbl:product_metrics}, where $k\geq 1$ and $m\geq 2$, and where a minus sign before one of the factors indicates the signature is taken to be mostly minus rather than mostly plus (e.g. $S^m$ will be positive definite, $-S^m$ is negative definite.). We restrict ourselves to Riemannian and Lorentzian factors in the table and we stress that the notation is only schematic since, unless $2k+1= m$, one of the factors needs to be rescaled to ensure the same Einstein constant for both factors. The standard $m$-dimensional hyperbolic space is denoted by $H^m$. \\

\begin{table}[h]
\centering
\begin{tabular}{|c|c|c|c|c|}
\hline
\textbf{Product metric $a$} & \textbf{Signature of $a$ $(+,-)$} & \textbf{Unit Killing vector $b$ % field
} & $\langle b,b\rangle$ & \textbf{Causal character} \\  \hline
$S^{2k+1} \times S^m$ & $(2k+1+m,0)$ & $b_{S^{2k+1}}$ & +1 & N/A \\ \hline
$S^{2k+1} \times -H^m$ & $(2k+1,m)$ &  $b_{S^{2k+1}}$ & +1 & N/A \\ \hline
$S^{2k+1} \times dS_m$  & $(2k+m,1)$ & $b_{S^{2k+1}}$ & +1 &spacelike  \\ \hline
$S^{2k+1} \times -AdS_m$ & $(2k+2,m-1)$ &  $b_{S^{2k+1}}$ & +1 & N/A\\ \hline
$AdS_{2k+1}\times -S^m$ & $(2k,m+1)$ & $b_{_{AdS_{2k+1}}}$ & -1 & N/A\\ \hline
$AdS_{2k+1}\times H^m$  & $(2k+m,1)$ & $b_{_{AdS_{2k+1}}}$ & -1 &timelike   \\ \hline
$AdS_{2k+1}\times -dS_m$ & $(2k+1,m)$ & $b_{_{AdS_{2k+1}}}$ & -1 &  N/A\\ \hline
$AdS_{2k+1}\times AdS_m$ &  $(2k+m-1,2)$ & $b_{_{AdS_{2k+1}}}$ &  -1 & N/A \\ \hline
\end{tabular}
\caption{Schematic overview of the possible product metrics and their unit Killing vector fields leading to Einstein-Kropina metrics, that can be formed from Riemannian and Lorentzian constant curvature spaces. A minus sign before one of the factors indicates the signature is taken to be mostly minus rather than mostly plus (e.g. $S^m$ will be positive definite, $-S^m$ is negative definite.) We stress that the overview is merely schematic; the second factor may need to be rescaled to ensure both factors have the same Einstein constant. }
\label{tbl:product_metrics}
\end{table}

\begin{rem}
    A straightforward combinatorial argument reveals that it is possible to construct $\kappa$-Einstein metrics with $\kappa\neq 0$, and a Killing field with $\langle b,b\rangle=1$, of every\footnote{ Except, of course, the case of a negative definite product metric ($p=0$). Applying a global $-$ sign produces the desired metrics with $\langle b,b\rangle=-1$ in all signatures except the positive definite one. One breaks this symmetry upon using Proposition \ref{prop:signature_Kropina_metric} to determine where the Einstein-Kropina metric has the relevant signature.
    
    % Here signature refers to the pseudo-Riemannian metric. As discussed above, if this is positive definite, then the resulting Kropina metric will also be positive definite; if this is Lorentzian and the vector field is timelike, then the Kropina metric will have Lorentzian signature within the timelike cone. More generally, the signature of the Kropina metric will be the same as that of the pseudo-Riemannian metric at all vectors $y$ that have the same causal character as $b$ \FV{With $\langle b,b\rangle > 0$} (see Prop. \ref{prop:signature_Kropina_metric}).
    
    } signature $(p,q)$ for which $p+q=n\geq 5$. This is doable via products of just Riemannian, negative definite and Lorentzian spaces; in fact, all the required combinations already appear in Table \ref{tbl:product_metrics} (perhaps up to global $-$ signs). Below, we explicitly detail the two most relevant ones in $n=5$. 
    
    Despite the flexibility to produce Einstein-Kropina metrics $L$, a notable fact is that the only one for which we can prove that $L$ is Lorentzian within the timelike cones, is that of  Example \ref{ex:ads_times_h}, corresponding to $AdS_{2k+1}\times H^m$ (and Example \ref{ex:ads} in the case of constant curvature spaces). In examples such as $S^{2k+1} \times dS_m$ (as per Proposition \ref{prop:signature_Kropina_metric}), the resulting  $L$ is only known to be Lorentzian within the set of \emph{spacelike} vectors of the product metric.

    This observation leaves open, e.g., whether there is an Einstein-Kropina metric that has Lorentzian signature within timelike cones and a \emph{negative} coefficient $\lambda=\frac{\kappa}{4}$.
\end{rem}

\begin{ex}[$S^3\times \tilde{S}^2$] 
Let $S^3$ be as in (\ref{n-sphere}), and $\tilde{S}^2= \big\{ x\in\R^{n+1}\,:\, \sum_{i=1}^{n+1} (x^i)^2=1/2\big\}$. The product metric $a_{S^3\times \tilde{S}^2}$ is Einstein with Einstein constant\footnote{We rescale the second factor to ensure equal Einstein constants for both factors; recall that the Einstein constant of a $n$-sphere of radius $R$ is $(n-1)/R^2$.} $\kappa = 2$. Consider also $b_{S^3}$ as in (\ref{KV n-sphere}). %, with $n=3$. 
  Then the positive definite Kropina metric $L_{S^3\times \tilde{S}^2} := A_{S^3\times \tilde{S}^2}^2/\beta_{S^3}^2$ is Einstein with Einstein coefficient $\lambda = 1/2$. 
\end{ex} 
\begin{ex}[$AdS_3\times \tilde{H}^2$] \label{ex:ads_times_h}
Let $\mbox{AdS}_3$ be as in (\ref{AdS_n}), with $n=3$, and $\tilde{H}^2$ hyperbolic space of radius $1/\sqrt{2}$ in $n=2$. The induced metric $a_{AdS_3\times \tilde{H}^2}$ is Einstein with Einstein constant $\kappa = -2$. Consider also $b_{AdS_3}$ as in (\ref{KV AdS_n}), with $n=3$. Then the Kropina metric $L_{AdS_3 \times \tilde{H}^2} = A_{AdS_3 \times \tilde{H}^2}^2/\beta_{AdS_3}^2$ is Einstein with Einstein constant\footnote{The overall sign difference with respect to the Einstein constant of $a$ is due to the signature flip $a\mapsto -a$, needed to apply our results (see footnote %\ref{footnote:signature_flip}).} 
\ref{footnote:sign_flip}).} $\lambda = 1/2$. This is a second example of Einstein-Kropina metric with indefinite, Lorentzian signature. 
\end{ex}

\subsubsection{Einstein-Sasaki-Kropina metrics} \label{Einstein-Sasaki}
 In Riemannian geometry, Einstein-Sasaki manifolds provide an important class of Einstein metrics, appearing naturally in e.g. supergravity and the AdS/CFT context \cite{Sasaki-Einstein-AdS/CFT}. It turns out that such metrics can be used to construct examples of Einstein-Kropina metrics in a very straightforward way in any odd dimension $\geq 5$.
 
 A Riemannian manifold $(M,a)$ is called \emph{Sasakian} if its metric cone $(C(M) := \R_{>0}\times M,\;\widehat a := \D r^2 + r^2a)$ admits an almost complex structure $J$ such that $(C(M), \widehat a, J)$ is K\"ahler (i.e., $\widehat a(J\cdot,J\cdot) = \widehat a(\cdot,\cdot)$ and $\nabla^{\widehat a} J = 0$). An \emph{Einstein-Sasaki} manifold is a manifold that is both Sasakian and Einstein. For a review, see \cite{Sasaki-Einstein}. Any Sasakian manifold admits a so-called \emph{Reeb vector field} $b := \left.J(r\partial_r)\right|_M$, which is Killing and satisfies $\langle b,b\rangle = 1$. Theorem \ref{theor:Krop_Einstein_cond} then has the following immediate implication.

 \begin{cor}\label{cor:Einstein-Sasaki}
     Suppose that $(M,a)$ is an Einstein-Sasaki manifold and $b$ is its Reeb field. Then the Kropina metric $L = A^2/\beta^2$ is Einstein.
 \end{cor}

\noindent Several classes of examples of Einstein-Sasaki manifolds are known \cite{Sasaki-Einstein-AdS/CFT,Sasaki-Einstein}, all of which lead to examples of Einstein-Kropina metrics, according to the corollary. Below, we explicitly describe one such class, originally introduced in \cite{Sasaki-Einstein_on_S3xS2} (see also \cite{Sasaki-Einstein}). In particular, \cite[Theorem 4.1]{Sasaki-Einstein} in combination with Corollary \ref{cor:Einstein-Sasaki} immediately yields the following novel class of Einstein-Kropina metrics. %\SH{compare to page 471 Besse}

\begin{theor}\label{theor:Einstein-sasaki}
There exist countably infinitely many Einstein-Kropina metrics $L_{p,q}$ with Einstein constant $
\lambda=1$ on $M = S^2\times S^3$, labelled by two positive integers $ q<p \in \mathbb Z_{>0}$ with $\text{gcd}(p, q) = 1$. Explicitly, in local coordinates, $L_{p,q}= A_{p,q}^2/\beta_{p,q}^2$ is given by
\begin{align}
   a_{p,q} &= \frac{1-y}{6}(\D\theta^2 + \sin^2\theta\,\D\phi^2) + \frac{1}{w(y)q(y)}\D y^2 + \frac{q(y)}{9}\left(\D\psi - \cos\theta\,\D\phi\right)^2 + w(y)
    \left[\D\alpha + f(y)\left(\D\psi - \cos\theta\,\D\phi\right)\right]^2, \label{eq:a_Einstein-Sasaki}\\
   b_{p,q} &= 3\partial_\psi - \tfrac{1}{2}\partial_\alpha= b, \label{eq:b_Einstein-Sasaki}
\end{align}
where\footnote{Note that the function $q(y)$ is not to be confused with the number $q$, nor our metric $a_{p,q}$ with the constant in \cite[Theorem 4.1]{Sasaki-Einstein}. We retain the notation of the original papers except where it conflicts with our conventions. } %We retain this notation to remain consistent with the original papers.}
\begin{align}
    w(y) = \frac{2(c-y^2)}{1-y},\qquad q(y) = \frac{c-3y^2+2y^3}{c-y^2}, \qquad f(y) = \frac{c-2y+y^2}{6(c-y^2)}
\end{align}
and \begin{align}
    c = \frac{1}{2} - \frac{p^2-3q^2}{4p^3}\sqrt{4p^2-3q^2}.
\end{align}
The Kropina metrics $L_{p,q}$ %all have Einstein constant 1 and 
are all positive definite on their maximal domains $\mathcal{A}_{p,q}^{\rm max}=T(S^2\times S^3)\setminus\{\beta_{p,q}=0\}$. %Countably many of them (if not all) are nonhomothetic. 
\end{theor}
\begin{proof}
   The result %\cite[Th. 4.1]{Sasaki-Einstein} 
   in \cite{Sasaki-Einstein} states that there exist countably infinitely many Einstein-Sasaki manifolds with Einstein constant $4$, labelled by $p,q$ as stated, where the Riemannian metric is locally given by \eqref{eq:a_Einstein-Sasaki}. In all cases, the Reeb vector field is given by \eqref{eq:b_Einstein-Sasaki}. Hence, their corresponding $L_{p,q}$ are Einstein by Corollary \ref{cor:Einstein-Sasaki} and must have Einstein constant $1$ by Theorem \ref{theor:Krop_Einstein_cond}.
\end{proof}

\noindent In the above theorem, $\theta\in[0,\pi],\phi\in[0,2\pi]$ are the standard coordinates on the first factor $S^2$.  The two coordinates $y\in[y_1,y_2]$ (where $y_1,y_2$ are the two smallest roots of the cubic $c-3y^2+2y^3$) and $\psi\in[0,2\pi]$ parameterize another $S^2$ (topologically). Then, $S^2\times S^3$ may be regarded as an $S^1$-fibration over the resulting $S^2\times S^2$, with $\alpha\in[0,2\pi l]$ the $S^1$-coordinate, where 
\begin{align}
    l = \frac{q}{3q^2 - 2p^2+p\sqrt{4p^2-3q^2}}.
\end{align}
The construction of Einstein-Sasaki metrics \cite{Sasaki-Einstein_on_S3xS2} used in the theorem can be generalized to higher dimensions \cite{Gauntlett:2004hh}, leading to the existence of infinitely many Einstein-Sasaki-Kropina metrics on higher analogs of $S^2\times S^3$. 

In addition, and in analogy to our earlier discussion, one can construct higher-dimensional Einstein-Kropina manifolds by taking the direct product of % the $a$ in 
\eqref{eq:a_Einstein-Sasaki} with other pseudo-Riemannian metrics with the same Einstein constant, i.e. $4$,so that the product is again Einstein. (The unit Killing vector being always the natural lift of \eqref{eq:b_Einstein-Sasaki}.) %Since the $b$ in \eqref{eq:b_Einstein-Sasaki} will still be unit Killing with respect to this product metric, the resulting Kropina metric will again be Einstein. 
For example, taking the product of %$a$
\eqref{eq:a_Einstein-Sasaki} with some de Sitter spacetime $dS_2$, one obtains  Einstein Lorentzian metrics $\widetilde a_{p,q}$ on $\widetilde M = S^2\times S^3\times dS_2\cong S^2\times S^3\times S^1\times\mathbb{R}$. Setting $\widetilde b$ just as in \eqref{eq:b_Einstein-Sasaki} yields infinitely many Einstein-Kropina metrics on the smooth manifold $S^2\times S^3\times S^1\times\mathbb{R}$ that are not %(everywhere) 
positive definite, according to Proposition \ref{prop:signature_Kropina_metric}. 

Finally, there is also the less %prominent 
explored notion of \emph{Lorentzian Einstein-Sasaki manifolds} (see e.g. \cite{axioms8010006}). For such manifolds, the Reeb vector field is also Killing and with $\langle b,b\rangle=\pm 1$, hence also leading to Einstein-Kropina metrics. However, no explicit (in the sense of Theorem \ref{theor:Einstein-sasaki}) examples of nontrivial Lorentzian Einstein-Sasaki manifolds seem to exist in the literature.

\section{$\Lambda$-vacuum solutions with arbitrary signature} \label{section: results vacuum eq}

For the rest of the article, we fix a $\Lambda\in\mathbb{R}$, the ``cosmological constant'', which is to be distinguished from the Einstein constant $\lambda$ of an Einstein-Kropina metric. We remark that all the discussion and results hold just the same in the case of a function $\Lambda(x)$.

Our next objective is to understand under what conditions an Einstein-Kropina metric solves the Finslerian field equation with cosmological constant in vacuum, \eqref{lambda vacuum eq}. To this end, we start with some useful expressions for the various geometric tensors of a general Kropina metric of Einstein type, as some of them simplify enormously. It is straightforward to obtain the simplified expressions using computer algebra software such as the Wolfram Mathematica package xAct \cite{MartinGarcia:xAct}, employing the identities in Appendix \ref{app:identities}. Since the computations and intermediate expressions are often extremely tedious, however, we omit them here. The following formulas hold whenever $L$ is an Einstein-Kropina metric written such that $\langle b,b\rangle=1$ (raising and lowering indices by $a$, as in Remark \ref{rem:raising_and_lowering}).

The spray coefficients \eqref{eq:spray_Krop} simplify to
\begin{align}
    G^i = \mathring\Gamma^i_{jk}y^jy^k - \frac{A}{\beta}s^i{}_jy^j.
\end{align}
The mean Cartan tensor attains the simple form
\begin{align}
    C_i = \frac{(n+1)\left(\beta y_i - A b_i\right)}{A\beta}.
\end{align}
%where we stress again that indices have been raised and lowered by $a$.
The canonical nonlinear connection is given by
\begin{align}
    N^k{}_i = \mathring\Gamma^k_{ij}y^j + \frac{A }{2\beta}s_i{}^k + \frac{A }{2\beta^2}b_is^k{}_{j}y^j - \frac{1}{\beta}y_is^k{}_jy^j.
\end{align}
The mean Landsberg tensor reduces to
\begin{align}\label{eq:mean_landsb_einst}
    \qquad P_i = - \frac{(n+1) s_{ij}y^j}{2\beta}.
\end{align}
Solely by virtue of the Einstein condition (namely $\text{Ric}(x,y)= \tfrac{\kappa}{4}L(x,y)$ as per Theorem \ref{theor:Krop_Einstein_cond}), the Ricci terms in the $\Lambda$-vacuum equation \eqref{lambda vacuum eq} can be written simply as
\begin{align}\label{eq:Ricci-part-EinsteinKrop}
    (n+2)\text{Ric} - g^{ij}\bar\partial_j\bar\partial_i\mathrm{Ric}\,L = (2-n)\frac{\kappa }{4}L.
\end{align}
(Recall that $\kappa$ is always the Einstein coefficient of $a$, i.e. $\mathring R_{ij} = \kappa\,a_{ij}$.) Finally, the Landsberg scalar \eqref{eq:Landsb_scal_0hom} of an Einstein-Kropina metric is given by
\begin{align}\label{eq:Landsb_scalar_Einstein_unit_KV}
    \mathcal P = -\frac{1+n}{4A}\left(2\kappa\beta^2 - \kappa A + 2n  s_{i}{}^{\ell}s_{j\ell}y^iy^j\right).
\end{align}

%\FV{As a (perhaps irrelevant) observation, for an Einstein-Kropina metric, solving the ($\Lambda$-)vacuum equation is equivalent to having isotropic Landsberg scalar (which then is $0$).} \\

Next we prove the following fundamental lemma.

\begin{lem}\label{lemma:main}
    Let $L=\left(\frac{A}{\beta}\right)^2$ be an Einstein-Kropina metric, expressed in such a way that $\langle b,b\rangle=1$. Suppose that the Landsberg scalar is isotropic, i.e. $\mathcal P(x,\y) = \mathcal P(x)$.
    Then:
    \begin{enumerate} [label=(\roman*)]
        \item $L$ is Ricci-flat (in other words, $a$ is Ricci-flat).
        \item $L$ is weakly weakly Landsberg, i.e., $\mathcal P =0$.
        \item $\nab_j b^k\nab_k b^i=0$.
    \end{enumerate}
\end{lem}
\begin{proof}
     The Landsberg scalar is given by \eqref{eq:Landsb_scalar_Einstein_unit_KV}, 
%i.e. \begin{align}\label{eq:Landsb_scalar_Einstein_unit_KV2}
%     \mathcal P = -\frac{1+n}{4A}\left(2\kappa\beta^2 - \kappa A + 2n  s_{i}{}^{\ell}s_{j\ell}y^iy^j\right).
% \end{align}
    and by assumption, we have $\mathcal P = \rho$ for some function $\rho:M\to \R$. This immediately implies that there is another $\tilde\rho:M\to \R$ such that
\begin{align}
    2\kappa\beta^2 + 2ns_{i}{}^{\ell}s_{j\ell}y^iy^j = \tilde\rho A.
\end{align}
Differentiating with respect to $y^i$ and $y^j$ yields
\begin{align}\label{eq:const_lands_scalar_cond2}
    2\kappa b_i b_j + 2n  s_{i}{}^{\ell}s_{j\ell} = \tilde\rho\,a_{ij}.
\end{align}
Now we may contract \eqref{eq:const_lands_scalar_cond2} with $b^ib^j$ to obtain
\begin{align}
    2\kappa\,\langle b,b\rangle^2 + 2n  s_{i}{}^{\ell}s_{j\ell}b^ib^j = \tilde\rho\,\langle b,b\rangle.
\end{align}
From here on, we use Appendix \ref{app:identities}. Since $b$ is Killing and $\langle b,b\rangle=1$ by assumption, $b^js_{j\ell}=0$ and the term $2n  s_{i}{}^{\ell}s_{j\ell}b^ib^j$ vanishes identically. It follows that 
\begin{equation}
    2\kappa  = \tilde\rho.
    \label{eq:kappa_and_tilde_rho}
\end{equation}
Next, we contract \eqref{eq:const_lands_scalar_cond2} with $a^{ij}$. This yields (again using that $\langle b,b\rangle=1$)
\begin{align}
    2\kappa + 2n  s^{j\ell}s_{j\ell} = n \tilde\rho = 2n\kappa.
\end{align}
Now we use the identity $s^{j\ell}s_{j\ell} = \kappa \langle b,b\rangle = \kappa$ (following from \eqref{eq:identity_unit_KV_4} under the Einstein condition for $a$). This yields
\begin{align}
    2\kappa + 2n  \kappa = n \tilde\rho = 2n\kappa,
\end{align}
or in other words, $\kappa = 0$, which means that both $a$ and $L$ are Ricci-flat (Theorem \ref{theor:Krop_Einstein_cond}).

Additionally, \eqref{eq:kappa_and_tilde_rho} asserts that $\tilde \rho = 0$. Substituting %$\kappa = \rho = 0$ 
$\kappa$ and $\tilde\rho$ into \eqref{eq:const_lands_scalar_cond2}, we obtain $s_{i}{}^{\ell}s_{j\ell} = 0$, or equivalently, $\mathring\nabla_j b^\ell\mathring\nabla_\ell b_i  = 0$. Putting everything together in \eqref{eq:Landsb_scalar_Einstein_unit_KV}, it follows that $\mathcal P = 0$, %i.e., $L$ is weakly weakly Landsberg. That concludes the proof of the lemma.
concluding the proof.
\end{proof}
%
% \FV{For me, $\mathcal{P}$ denotes the \emph{0-homogeneous} mean Landsberg terms in the field equation and $\nab$ denotes the Levi-Civita connection of $a$. \\
% If we don't feel very comfortable thinking about cosmological \emph{functions} and such, we may state this with $\rho$ being a constant. But I think the current formulation might have a couple advantages. First, it makes clear that the only thing required is that $\mathcal{P}(x,\y)=\mathcal{P}(x)$. Second, it might even apply to the $2D$ case, where Einstein ``constant'' $\lambda$ is really a function $\lambda(x)$ (worth checking?).} \\

% \begin{cor}
%     For all Einstein Kropina metrics (of all signatures) that solve Pfeifer and Wohlfarth's vacuum equation, 
%     \[
%     \mathrm{Ric}=0.
%     \]
% \end{cor}
%
\noindent The following theorem characterizes Einstein-Kropina solutions to the $\Lambda$-vacuum equation \eqref{lambda vacuum eq}, in terms of the building blocks $a$ and $b$. Surprisingly, it shows that the cosmological constant must necessarily be $0$. 

\begin{theor}\label{general vacuum solutions}
    Assume that $n\left(=\dim M\right)>2$. A Kropina metric $L=\left(\frac{A}{\beta}\right)^2$ of arbitrary signature with (w.l.o.g.) $\langle b,b\rangle=1$ is a solution to \eqref{lambda vacuum eq} of Einstein type if and only if the following four conditions are satisfied:
    \begin{itemize}
        \item $a$ is Ricci-flat,
        \item $b$ is Killing  with respect to $a$,
        \item $\nab_j b^k\nab_k b^i=0$,
        \item $\Lambda = 0$.
    \end{itemize}
    In that case, $L$ is Ricci-flat and weakly weakly Landsberg. % i.e. $\mathcal P = 0$.
\end{theor}
\begin{proof}
    For the first implication, assume that $L$ is a $\Lambda$-vacuum solution of Einstein type. First, since $L$ is Einstein and $\langle b,b\rangle=1$, it follows by Theorem \ref{theor:Krop_Einstein_cond} that $b$ is Killing. Second, as noted in \eqref{eq:Ricci-part-EinsteinKrop}, for any  Einstein-Finsler metric,  the Ricci terms of the field equation add up to a contribution of the form $\rho(x) L(x,y)$ for some $\rho\in C^\infty(M)$. Hence, the full %field equation reads 
    \eqref{lambda vacuum eq} reads
    \begin{align}
        \rho L - L\mathcal P +2\Lambda\,L= 0,
    \end{align}
     showing after division by $L$ %(and continuous extension to the set where $L=0$) 
     that $\mathcal P(x,y)= \rho(x)+2\Lambda$. In this situation, Lemma \ref{lemma:main} applies and we conclude that $a$ and $L$ are Ricci-flat (from where $\rho = 0$), $\nab_j b^k\nab_k b^i=0$, and $0=\mathcal P = 2\Lambda$. \\ %This proves the first implication.\\
    For the second implication, suppose that the four conditions stated hold. First of all, since $b$ is Killing and $a$ is Ricci-flat, it follows %by Theorem \ref{theor:Krop_Einstein_cond} 
    that $L$ is Ricci-flat. Hence, the Ricci part of the field equation vanishes identically and the full vacuum eq. becomes equivalent to $\mathcal P = 0$. However, $\mathcal P$ is given in this situation by \eqref{eq:Landsb_scalar_Einstein_unit_KV} with $\kappa = 0$ and $s_{i}{}^{\ell}s_{j\ell} = 0$ (the latter being simply our condition $\nab_j b^k\nab_k b^i=0$). In this way,
    \begin{align}
        \mathcal P = \mathrm{Ric}=0 
    \end{align} 
    identically and $L$ is a  solution to \eqref{lambda vacuum eq}.
 \end{proof}

 \noindent In what follows, we will see that when $a$ is either Riemannian or Lorentzian, the condition $\nab_j b^k\nab_k b^i=0$ in %Theorem \ref{general vacuum solutions} 
 the above result can be reduced to $\nab_j b^i=0$. The argument we use does not generalize to other signatures, thus allowing in principle for the possibility %that in certain signatures there could exist 
 of the existence of pairs $(a,b)$ (with $a$ Ricci-flat and $b$ unit Killing)  satisfying $\nab_j b^k\nab_k b^i=0$ but not $\nab_i b^j=0$. 
 
 Interestingly, such pairs would produce Einstein-Kropina metrics that are weakly weakly Landsberg (by the above theorem) but not weakly Landsberg (by \eqref{eq:mean_landsb_einst}) and so, in particular, not Berwald. In the spirit of the famous unicorn problem \cite{Bao2007}, one could call such metrics  \textit{weak-weak unicorns}. It is an open question whether weak-weak unicorns exist at all. The results that follow below imply that there exist no weak-weak unicorns within the class of Einstein-Kropina metrics in Riemannian or Lorentzian signature (meaning the signature of $a$, not of $L$, which need not be the same).

\section{Classification of vacuum solutions in Riemannian and Lorentzian signatures} \label{section: classification vacuum solutions}

In Riemannian and Lorentzian signatures, the conditions in Theorem \ref{general vacuum solutions} turn out to be very restrictive. To better understand this, we will regard the $(1,1)$-tensor $\nab_jb^i$ as a field of endomorphisms $J$ of the tangent spaces: for $x\in M$,
\begin{align}
    J\colon T_xM\longrightarrow T_xM,\qquad J(X):= \left.X^j\nab_j b^i(x)\,\partial_i\right|_x.
\end{align}
The Killing condition for $b$ translates into the skew-adjointness of $J$: for $X,Y\in T_xM$,
\begin{align}
    a(J(X),Y)+a(X,J(Y))=0.
\end{align}
%\SH{do we need what follows below until the lemma? otherwise we can remove it.}
%from where
%a(J(X),X)=0,\qquad \mathrm{trace}(J)\left(=\mathring{\mathrm{div}}\,b\right)=0.
%\]
%So, $J$, and then $\nab_j b^i$, is determined by its $a$-antisymmetric part:
%\[
%a(J(X),Y)=\frac{1}{2}\left(a(J(X),Y)-a(X,J(Y)\right)=\frac{1}{2}\left(X^j\nab_j b_i(x)\,Y^i-X^j\nab_i b_j(x)\,Y^i\right)=\frac{1}{2}\left(\D \beta\right)_x(X,Y)
%\]
%(regarding $\beta=b_i\,\D x^i$ as a $1$-form on $M$ and following the conventions of \cite{lee} for the exterior differential).\\
%
Moreover, the third condition in Theorem \ref{general vacuum solutions} can be written as $J^2=0$. 

\begin{lem}\label{lem:J^2=0impliesJ=0}
    If $a$ is Riemannian or Lorentzian, then any Killing vector field $b$ that satisfies $\nab_j b^k\nab_k b^i=0$ is parallel, $\nab_k b^i=0$.
\end{lem}
\begin{proof}
    In the notation introduced above, the conditions of the lemma imply that $J$ is skew-adjoint and satisfies $J^2=0$, and we have to show that $J=0$.
    
    Since $J$ is skew-adjoint, $J^2=0$ implies that $a(J(X),J(X)) = -a(X,J^2(X)) = 0$ for any $X\in T_xM$. In other words, the image of $J$, denoted $\mathrm{im}(J)$, is null for $a$. In Riemannian signature, this already implies that $J = 0$; in the Lorentzian case, $\mathrm{im}(J)$ is a linear subspace contained in the null cone. This means that it must be at most $1$-dimensional and, hence, $J$ is given by $J(X) = 
    \phi(X)\,Z$ for some %reference 
    vector $Z\in T_x M\setminus\left\{0\right\}$ and linear form $\phi\in T^\ast_x M$. %We will assume w.l.o.g. that $Z\neq 0$, for otherwise the result follows immediately. 
    Now we use again the skew-adjointness of $J$. We have, for all $X,Y\in T_xM$, that $\phi(X) a(Z,Y) = -\phi(Y)a(X,Z)$. Choosing any $X$ such that $a(X,Z)\neq 0$ and setting $Y=X$, we obtain $\phi(X) = -\phi(X)$, hence $\phi(X) = 0$, so $\phi$ vanishes except at most in the set %given by $a(X,Z)=0$
    $a$-orthogonal to $Z$. However, since this set is a linear subspace %with a dimension strictly less then $\dim M$ (since $Z\neq 0$),
   of dimension $n-1=\dim T_xM-1$, its complement is dense and it  follows, by continuity, %of $J$ and hence $\phi$ 
    that $\phi=0$ and $J=0$ identically.
\end{proof}

% \begin{rem}
%     It is well known that the vanishing of $\nab b$ induces a local splitting of $a$ as the orthogonal sum of a $1$-dimensional metric and a $\left(n-1\right)$-dimensional one. A concrete construction of it goes as follows. Since $\beta$ is closed ($\D\beta=0$), there exists a function $x^1$ around each point $x_0\in M$ such that $\D x^1=\beta$. Calling $x^1_0:=x^1(x_0)$ and $\phi_t$ the flow of the vector field $b$, the map
% \[
% \Phi\colon \left\{(t,p)\in \R\times M\colon x^1(p)=x^1_0,\;\phi_t(p) \hbox{ \rm exists}\right\}\longrightarrow M,\qquad\Phi(t,p)=\phi_t(p)
% \]
% can be inverted around $x_0=\Phi(0,x_0)$. This allows us to identify some neighborhood of $x_0$ with $I\times\Sigma$, where $I$ is a manifold contained in the hypersurface $\left\{x^1=x^1_0\right\}$ of $M$, the product being equipped with the metric $\Phi^\ast a$. Now, a couple straightforward identities at any $(t,p)\in I\times\Sigma$, $\partial_t$ being the canonical vector field on $I$ and $Z\in T_p\Sigma$:
% \[
% \D\Phi_{(t,p)}(\partial_t)=b_{\phi_t(p)}=\D(\phi_t)_p(b_p),\qquad \D\Phi_{(t,p)}(Z)=\D(\phi_t)_p(Z).
% \]
% \end{rem}

\noindent It is well known that the vanishing of $\mathring\nabla b$ with $b$ nonzero induces a local splitting of $a$ as a product of a $1$-dimensional metric and a $\left(n-1\right)$-dimensional one. That is to say, we have the adapted orthogonal decomposition
\begin{align}
    a=(\D x^1)^2+\sum_{k,\ell=2}^n \bar{a}_{k\ell}(x^2,\ldots,x^n)\,\D x^k\otimes\D x^\ell,\qquad b=\partial_{x^1},
\end{align}
where we are also using the fact that $\langle b,b\rangle = 1$; a proof  can be found, e.g., in \cite[Proposition A.2.2]{Heefer:2024Thesis}. Conversely, for a metric that is locally of this form, $b = \partial_{x^1}$ is parallel, as is straightforward to check.

\begin{theor} \label{riem and lor classification}
   Assume that $n>2$. Let $L=\left(\frac{A}{\beta}\right)^2$ be an Einstein-Kropina metric %with $\dim M>2$ 
    expressed in such a way that $\langle b,b\rangle=1$ and suppose that $a$ is either Riemannian or Lorentzian. Then, the following are equivalent:
    \begin{enumerate} [label=(\roman*)]
        \item \label{riem and lor classification 1} $L$ solves the $\Lambda$-vacuum equation \eqref{lambda vacuum eq}.

        \item \label{riem and lor classification 2} $L$ is weakly weakly Landsberg.

        \item \label{riem and lor classification 3} $L$ is Berwald.

        \item \label{riem and lor classification 4} $\nab_j b^i=0$.

        \item \label{riem and lor classification 5}  %Locally, $a$ is a product metric on a line times a $\left(n-1\right)$-dim. manifold and $b$ is the corresponding parallel field:
        %\[
        %a=\D x^1\otimes\D x^1+\sum_{k,l=2}^n \hat{a}_{kl}(x^2,\ldots,x^n)\,\D x^k\otimes \D x^l,\quad b=\partial_{x^1}.
        %\]
        Locally, $a$ has the form
        \begin{align}
            a=(\D x^1)^2+ \sum_{k,\ell=2}^n \bar{a}_{k\ell}(x^2,\ldots,x^n)\,\D x^k\otimes\D x^\ell
        \end{align}
        with $b=\partial_{x^1}$ and $\bar{a}_{k\ell}$ the components of an $(n-1)$-dimensional pseudo-Riemannian metric.
    \end{enumerate}
    In this case, $L$, $a$ and $\bar{a}$ are all Ricci-flat, 
    \begin{align}
        \Lambda = 0,
    \end{align}
    and the signature of $\bar{a}$ is:
    \begin{itemize}
        \item $(+ \ldots +)$ if $a$ is Riemannian,
        \item $(-\ldots -)$ if $a$ is Lorentzian and $b$ is timelike,
        \item $(- +\ldots +)$ is $a$ is Lorentzian and $b$ is spacelike.
    \end{itemize}
\end{theor}
\begin{proof}
    The equivalence \ref{riem and lor classification 4}$\Longleftrightarrow$\ref{riem and lor classification 5} is what we have remarked above. Regarding the others, it thus suffices to prove that \ref{riem and lor classification 4}$\implies$\ref{riem and lor classification 3}$\implies$\ref{riem and lor classification 2}$\implies$\ref{riem and lor classification 1}$\implies$\ref{riem and lor classification 4}.  Once this is done, that $L$ and $a$ are Ricci-flat will follow from Lemma \ref{lemma:main}, and Ricci-flatness of $\bar a$ follows from that of $a$ by the well-known decomposition of the Ricci tensor for product metrics. The statements about the signature of $\bar a$ are immediate from the signature of $a$ and the corresponding meaning of the sign of $a(b,b)=\langle b,b\rangle$. 
    
    We now prove the implications \ref{riem and lor classification 4}$\implies$\ref{riem and lor classification 3}$\implies$\ref{riem and lor classification 2}$\implies$\ref{riem and lor classification 1}$\implies$\ref{riem and lor classification 4}, some of which are trivial. Clearly, \ref{riem and lor classification 4}$\implies$\ref{riem and lor classification 3} since it is known, e.g. \cite[Corollary 5.2.2]{Heefer:2024Thesis}, that any $(\alpha,\beta)$-metric with a parallel $b$ has as its canonical connection the Levi-Civita one of $a$ (in particular, it is Berwald). Of course, \ref{riem and lor classification 3}$\implies$\ref{riem and lor classification 2}: any Berwald metric  is Landsberg, hence weakly weakly Landsberg. For \ref{riem and lor classification 2}$\implies$\ref{riem and lor classification 1}, one applies Lemma \ref{lemma:main} with  $\mathcal{P}(x,y)=0$, obtaining that $L$ is Ricci-flat and hence a vacuum solution to Pfeifer and Wohlfarth's equation (i.e., $\Lambda=0$). Finally, for \ref{riem and lor classification 1}$\implies$\ref{riem and lor classification 4}, it follows by Theorem \ref{general vacuum solutions} that $\nab_j b^k\nab_k b^i=0$, and by Lemma \ref{lem:J^2=0impliesJ=0}, that $\nab_j b^i=0$.
\end{proof}

\noindent As an important consequence, if $n=3$ or $4$, all $\Lambda$-vacuum solutions are (locally) trivial. 

\begin{cor}\label{cor:solutions_n=3,4}
    %FIRST FORMULATION. Suppose the dimension of $M$ is $3$ or $4$ and the Einstein Kropina metric $L$ be with Lorentzian (resp., Riemannian) $a$ and $\langle b,b\rangle=1$. Then, $L$ solves \br reference vacuum eq. \er if and only if $(M,a)$ is locally isometric to Minkowski spacetime (resp., Euclidean space) and, under such identification, $b$ is a generator of unit translations.

    %SECOND FORMULATION If, moreover, the dimension of $M$ is $3$ or $4$ then $L$ is a solution if and only if $(M,a)$ is locally isometric to Minkowski spacetime or Euclidean space and, under such identification, $b$ is a generator of unit translations.

    %MY LAST PROPOSAL:
    When $\dim M$ is $3$ or $4$, the Einstein-Kropina metric $L$, with $a$ Riemannian or Lorentzian and $\langle b,b\rangle=1$, solves \eqref{lambda vacuum eq} if and only if $(M,a)$ is locally isometric to Euclidean space or Minkowski spacetime. Under this identification, $b$ must be a constant (traslational) vector field.
\end{cor}
\begin{proof}
    This follows immediately from Theorem \ref{riem and lor classification}, because in this case, $\bar a$ is a Ricci-flat metric in dimension $2$ or $3$, which implies it is flat and so is $a$.
\end{proof}

\noindent Lastly, since all Einstein-Kropina solutions of the $\Lambda$-vacuum equation are Ricci-flat, by Theorem \ref{general vacuum solutions}, it is natural to ask whether Ricci-flatness is, in fact, a sufficient condition to be a vacuum solution. The following result gives a partial answer to this question.

\begin{theor}\label{theor:ricci-flat_sufficient_cond}
     Let $L=\left(\frac{A}{\beta}\right)^2$ be a Kropina metric and suppose that either $a$ is Riemannian or $a$ is Lorentzian with timelike $b$. If $\mathrm{Ric}=0$, then $L$ solves Pfeifer and Wohlfarth's vacuum eq. \eqref{vacuum eq} (hence it is described by Theorem \ref{general vacuum solutions} and Theorem \ref{riem and lor classification}).
\end{theor}
\begin{proof}
    We may assume w.l.o.g. that $\langle b,b\rangle=b_i b^i = 1$ by rescaling $a$, as explained above. Since $L$ is Ricci-flat (so Einstein), Theorem \ref{theor:Krop_Einstein_cond} guarantees that $b$ is Killing and $a$ is Ricci-flat. The lemma below shows that this implies that $b$ is parallel. Then, Theorem \ref{riem and lor classification} implies that $L$ is a solution, as desired. 
\end{proof}

\begin{lem}
    Let $a$ be a Ricci-flat pseudo-Riemannian metric and $b$ be a Killing field of constant (pseudo-)norm. If $a$ is Riemannian or $a$ is Lorentzian with timelike $b$, then $b$ is parallel with respect to $a$. 
\end{lem}
\begin{proof}
    Put, w.l.o.g., $\langle b,b\rangle=1$. A contraction of the identity \eqref{eq:2nd_cov_der_of_Killing_field} yields $\mathring\nabla^i\mathring\nabla_ib_k = -\mathring R^\ell{}_{k}b_\ell = 0$. Using this,
    \begin{align}
        0 = \mathring\nabla^i\mathring\nabla_i \left(b^jb_j\right) = 2\left( b^j\mathring\nabla^i\mathring\nabla_i b_j  + \mathring\nabla^i b^j\mathring\nabla_i b_j\right) = 2\mathring\nabla^j b^i\mathring\nabla_j b_i, 
    \end{align}
    where we have used the metric-compatibility of the Levi-Civita connection. Now introduce a local orthonormal frame $\{E_c\}_{c=1,\dots, n}$ with $E_1 = b$, so that $a_{ij}E^i_cE^j_d = \varepsilon_c\delta_{cd}$ with $\varepsilon_1 = 1$ and $\varepsilon_c = \pm 1$ for $c=2,\dots,n$, the sign being positive (resp., negative) in Riemannian (resp., Lorentzian) signature. Evaluating the above expression in this local orthonormal frame, we find
    \begin{align}\label{eq:sum}
        0=\mathring\nabla^j b^i\mathring\nabla_j b_i =  \sum_{c=2}^n \varepsilon_c a(\mathring\nabla_{E_c} b,\mathring\nabla_{E_c} b) = \pm\sum_{c=2}^n a(\mathring\nabla_{E_c} b,\mathring\nabla_{E_c} b).
    \end{align}
    The $c=1$ term is absent as $\mathring\nabla_{E_1} b = \mathring\nabla_b b = 0$ because of the identity \eqref{eq:identity_unit_KV_5} for unit Killing fields. Furthermore, the vectors $\mathring\nabla_{E_c} b$ are all orthogonal to $E_1=b$, because $b^i\mathring\nabla_{E_c} b_i = \tfrac{1}{2}\mathring\nabla_{E_c}(b^i b_i) = 0$ by constancy of the norm. Hence, the vectors $\mathring\nabla_{E_c} b$ lie in the span of $E_2,\dots,E_n$, which is positive  (resp., negative) definite. This means all terms in the sum \eqref{eq:sum} have the same sign and the vanishing of the sum implies the vanishing of each of the individual terms, that is, $a(\mathring\nabla_{E_c} b,\mathring\nabla_{E_c} b)=0$ for each $c=2,\dots, n$. Using again the fact that $a$ is positive (resp., negative) definite on the subspace in which the $\mathring\nabla_{E_c} b$  belong,  this implies that $\mathring\nabla_{E_c} b=0$ for each $c=2,\dots, n$. We already knew that $\mathring\nabla_{E_1} b = \mathring\nabla_b b = 0$, so this simply says that $\mathring\nabla b = 0$, i.e., $b$ is parallel.
\end{proof}
%\FV{This uses the fact that $\overline{\mathrm{Ric}}(b,b)=\sum g(\nab_{E_i} b,\nab_{E_i} b)$, where the vectors $\nab_{E_i} b$ live in a Euclidean subspace. So, if $\mathrm{Ric}=0$, all these vectors have to be $0$, yielding $\nab b=0$.} 

%\SH{Observe/recall that in 4D any ricci flat with timelike unit killing is flat.}

\subsection{An explicit example in $n=5$}
Corollary \ref{cor:solutions_n=3,4} shows that in dimensions $3$ and $4$, the %$\Lambda$-
vacuum equation has only trivial solutions (up to global topology) in the Einstein-Kropina class. To produce interesting examples of solutions, we thus have to start in dimension $5$. Theorem \ref{riem and lor classification} shows that for any Einstein-Kropina $\Lambda$-vacuum solution, we have a local splitting $a = (\D x^1)^2 + \bar a$ with $\bar a$ being a Ricci-flat $(n-1)$-dimensional metric. But conversely, such a splitting immediately implies % (since $a$ is Einstein) 
that $a$ is Ricci-flat and $b:=\partial_{x^1}$ is parallel (hence Killing), so  $L:=L_{a,b}$  is Einstein and, by the theorem, it solves the field equation. So, we have the following corollary of Theorem \ref{riem and lor classification}. 
\begin{cor}
    Let $L=\left(\frac{A}{\beta}\right)^2$ be a Kropina metric expressed such that $\langle b,b\rangle=1$ and suppose that $a$ is Riemannian or Lorentzian. Then, $L$ is a %$\Lambda$-
    vacuum solution of Einstein type if and only if, locally, $a$ is a product metric  $a = (\D x^1)^2 + \bar a$, with $b=\partial_{x^1}$ and $\bar a$ some Ricci-flat $(n-1)$-dimensional metric.
\end{cor}
\noindent This result tells us exactly how to construct new examples of vacuum solutions. 

\begin{ex}
    Let $\bar a$ be the Euclidean Schwarzschild metric (see e.g. \cite{Hawking:1976jb} and references therin) in geometrized units ($c=G=1$):
    \begin{align}\label{eq:_Eucl_SS_metric}
        \bar a =  \left(1-\frac{2 M}{ r}\right)\D t^2 + \left(1-\frac{2 M}{ r}\right)^{-1}\D r^2 + r^2\left(\D \theta^2 +\sin^2\theta\,\D\phi^2\right) \qquad (r>2M).
    \end{align}
    This metric is positive definite and Ricci-flat; equivalently, $-\bar a$ is negative definite and Ricci-flat. According to the above corollary, the pair
    \begin{align}
        a &= \D x^2 + (-\bar a) = \D x^2 - \left(1-\frac{2 M}{ r}\right)\D t^2 - \left(1-\frac{2 M}{ r}\right)^{-1}\D r^2 - r^2\left(\D \theta^2+\sin^2\theta\,\D\phi^2\right),\\
        b &= \partial_x
    \end{align}
    defines a (Ricci-flat) Einstein-Kropina metric $L = A^2/\beta^2$ that solves the $\Lambda$-vacuum equation with $\Lambda=0$. Since $a$ is Lorentzian and $b$ is timelike, $L$ has Lorentzian signature within the timelike cones, by Proposition \ref{prop:signature_Kropina_metric}.
\end{ex}

\section{Discussion}

In this work, we studied Einstein-Kropina metrics and their role as solutions to Pfeifer and Wohlfarth's vacuum equation in Finsler gravity, including a possibly nonvanishing cosmological constant (or function). Our first goal was to extend the characterization of Einstein-Kropina metrics obtained in the positive definite setting by Zhang, Shen and others to arbitrary signatures. We showed that the same result continues to hold: a Kropina metric $L=(A/\beta)^2$ with (w.l.o.g.) %$\langle b,b\rangle=1$
$a(b,b)=1$ is Einstein if and only if the underlying pseudo-Riemannian metric $a$ is Einstein and the vector field $b$ is Killing. This allows for the construction of examples of Einstein-Kropina metrics from pseudo-Riemannian Einstein manifolds admitting a unit Killing field. Using this characterization, we constructed various explicit Einstein-Kropina metrics in both Riemannian and Lorentzian signatures, and we showed that they exist in arbitrary signature $(p,q)$ with $p+q=n\geq 5$.\\

Our second goal was to determine which Einstein-Kropina metrics solve the $\Lambda$-vacuum equation \eqref{lambda vacuum eq}. The resulting classification turns out to be remarkably restrictive. We showed that any Einstein-Kropina solution of the $\Lambda$-vacuum equation must in fact be Ricci-flat and weakly weakly Landsberg (if $a$ is Riemannian or Lorentzian, it must even be Berwald). The other surprise was that the cosmological constant necessarily vanishes: in the present framework, Einstein-Kropina metrics are incompatible with the presence of a nonvanishing $\Lambda$. In this sense, the theory dynamically enforces $\Lambda=0$ within the Einstein-Kropina class (in stark contrast to Einstein metrics in general relativity). An important question that arises naturally is whether there exist (properly Finslerian) solutions to the gravitational field equation with $\Lambda\neq 0$ at all: currently, the answer is unknown.\\

In the Riemannian and Lorentzian cases, the constraints become even stronger. The condition $\mathring\nabla_j b^k \mathring\nabla_k b^i=0$ forces the Killing vector field $b$ to be parallel, implying that the underlying metric $a$ locally splits as a product of a one-dimensional factor and a Ricci-flat metric of one dimension lower. Consequently, Einstein-Kropina vacuum solutions are locally determined by Ricci-flat metrics of dimension $n-1$. In particular, in dimensions $n=3$ and $n=4$, this implies that all solutions are flat. In $n=2$, no nontrivial Einstein-Kropina metrics exist to begin with, independently of the field equation. In this way, nontrivial examples only occur for $n\geq5$.\\

Several questions remain open. One concerns the existence of nontrivial Einstein-Kropina metrics in dimension $4$. %Preliminary results suggest that this case is particularly constrained, and 
Taking into account the constraints, it remains an interesting problem to determine whether such metrics exist at all. Another open question concerns the possibility of (necessarily Ricci-flat) Einstein-Kropina solutions to the %$\Lambda$-
vacuum equation in more general signatures, where the condition $\mathring\nabla_j b^k \mathring\nabla_k b^i=0$ does not appear to imply that $b$ is parallel. Interestingly, such cases would lead to examples of Kropina metrics that are weakly weakly Landsberg but not Berwald, and are therefore directly relevant to the Landsberg unicorn problem. Finally, as mentioned above, it would be natural to investigate solutions to the $\Lambda$-gravity equation in more general classes of pseudo-Finsler metrics, beyond the Einstein-Kropina setting.

\section*{Acknowledgements}
We would like to acknowledge Prof. Nicoleta Voicu for valuable correspondence, and students Marta Amaro Calatayud and Francesco Botta for their insights on topics closely related to this work.  

FFV was partially supported by the project PID2024-156031NB-I00 funded by MICINN/AEI and FEDER.

\appendix

\section{List of identities}\label{app:identities}

Below, some standard identities are listed for a vector field $b$ with certain properties. For the definition of the symbols and other conventions, we refer to the main text.\\

\noindent \textbf{Constant norm}\\
If $b$ has constant (pseudo-)norm, i.e. $\langle b,b\rangle \equiv a(b,b) = constant$, then
\begin{align}\label{eq:identity_unit_norm}
    b^i\mathring\nabla_jb_i = 0.
\end{align}

\noindent \textbf{Killing vector}\\
If $b$ is a Killing vector, then
\begin{align}
    \mathring\nabla_ib_j + \mathring\nabla_jb_i &= 0,\\
    \mathring\nabla_i\mathring\nabla_jb_k &= \mathring R^\ell{}_{ijk}b_\ell.\label{eq:2nd_cov_der_of_Killing_field}
\end{align}

\noindent \textbf{Killing vector with constant norm}\\
%If $b$ is a Killing vector with constant norm, i.e. $\langle b,b\rangle \equiv a(b,b) = constant$,
If both of the above conditions hold, then
\begin{align}
    \mathring\nabla_ib_j + \mathring\nabla_jb_i &= 0,\\
    \mathring\nabla_i\mathring\nabla_jb_k &= \mathring R^\ell{}_{ijk}b_\ell, \label{eq:identity_unit_KV_6} \\
    b^i\mathring\nabla_jb_i = b^i\mathring\nabla_ib_j &= 0,\label{eq:identity_unit_KV_5} \\
    b^is_{ik} = b^is_{ki} &= 0, \label{eq:identity_unit_KV_1}\\
    b^k\mathring\nabla_is_{jk} &= -s_{kj}s_i{}^k,   \label{eq:identity_unit_KV_2} \\
    \mathring\nabla^ks_{ik} &= -\mathring R_{ki}b^k, \label{eq:identity_unit_KV_3}\\
     s_{ij}s^{ij} &= \mathring R_{ij} b^i b^j.   \label{eq:identity_unit_KV_4} 
\end{align}

\section{Proof of the Einstein condition}\label{app:proof_Einstein_cond}

We sketch the proof of Theorem \ref{theor:Krop_Einstein_cond}, repeated below for clarity. Apart from several small modifications (especially in the 2D case), it is very similar to the proof given by Zhang and Shen \cite{ZHANG201380} in the positive definite setting.
\renewcommand{\theprop}{\ref{theor:Krop_Einstein_cond}}
\begin{theor}[Einstein condition for Kropina metrics]
    A Kropina metric $L=\left(\frac{A}{\beta}\right)^2$ with $n\geq 2$ and written such that $\langle b,b\rangle=1$ is Einstein if and only if $a$ is Einstein and $b$ is Killing. In this case, if $\mathring R_{ij} = \kappa\,a_{ij}$ then $\mathrm{Ric} = \tfrac{\kappa}{4}L$.
\end{theor}
\begin{proof}
Assume w.l.o.g. that $L=A^2/\beta^2$ with $\langle b,b\rangle = 1$. Then it can be shown by explicit computation that\footnote{Note that $A,\beta\neq 0$ on the domain of definition $\mathcal A$ of the Kropina metric. That $\beta \neq 0$ is clear from the definition of the Kropina metric; the condition $A\neq 0$ follows from the fact that the fundamental tensor would degenerate at $A=0$, by \eqref{eq:fund_tensor_Krop}.}
\begin{align}\label{eq:_Ricci_Krop_unit}
    \text{Ric} = \frac{1}{A^2\beta^2}\left(3(n-1)\beta^4\left(r_{ij}y^i y^j\right)^2 + A p\right),
\end{align}
where $p:\mathcal A\subset TM \to \R$ is some function that is in particular polynomial in $y$. 

For the first implication, suppose that $L$ is Einstein. We will show that $a$ must be Einstein and $b$ Killing. The Einstein condition $\text{Ric}(x,y) = \lambda(x)\,A(x,y)^4/\beta(x,y)^2$ can be written, multiplying both sides by $A^2 \beta^2$, as
\begin{align}\label{eq:Einstein_cond_Krop}
    3(n-1)\beta^4\left(r_{ij}y^i y^j\right)^2 = A (\lambda A^5-p) 
\end{align}

In what follows, we will use the standard result that the ring of multivariate polynomials over any field is a unique factorization domain (UFD), which means in particular that if a product of polynomials is divisible by some irreducible polynomial, then one of the factors must be divisible by said irreducible polynomial. We distinguish two cases, $n\geq 3$ and $n=2$.\\

\noindent\textbf{The case $\bm{n\geq 3}$}\\
We start by observing that if $n \geq 3$, then $A$ is an irreducible polynomial in $y$. For if it were not, then we could write $A = ql$ for two $y$-linear functions $q = q_i(x)y^i$ and $l = l_i(x)y^i$, and taking two vertical derivatives would yield $2a_{ij} = q_il_j + l_iq_j$, showing that $a$ has at most rank 2, which would contradict nondegeneracy if\footnote{In the positive definite setting, $2a_{ij} = q_il_j + l_iq_j$ would still provide a contradiction  %contradict nondegeneracy 
if $n=2$, which is why a separate argument is not necessary in that case. However, in $n=2$, %there do exist 
Lorentzian metrics % which 
can be factored as a product of two linear forms; in that case, there is not yet a contradiction. An example is the Minkowski metric, with quadratic form %$\D s^2 = (\D x^1 + \D x^0)(\D x^1 - \D x^0)$.
$ds^2=\left(y^1 + y^0\right)\left(y^1-y^0\right)$: it can be put as $ q_il_j + l_iq_j$. %That is why in the indefinite case we need to discuss the $n=2$ case separately.
} $n\geq 3$.

Since the RHS of \eqref{eq:Einstein_cond_Krop} is divisible by $A$, the LHS must be as well, and by the above reasoning, this means that either $\beta$ or $r_{ij}y^i y^j$ must be divisible by $A$. Since $\beta$ is clearly not divisible by $A$, it follows that there exists some $\rho:M\to\R$ such that
\begin{align}
   r_{ij}(x)\,y^i y^j = \rho(x) A(x,y).
\end{align}
Taking two vertical derivatives yields $r_{ij} = \rho a_{ij}$, which says that $b$ is conformally Killing. Since $\langle b,b\rangle=1$, this furthermore implies that $b$ must be Killing, i.e., $r_{ij} = 0$ and $\rho = 0$ (this follows by contracting the conformal Killing equation with $b^ib^j$ and applying \eqref{eq:identity_unit_norm}). Substituting the latter condition into \eqref{eq:Einstein_cond_Krop}, multiplying both sides by $A^2$, and employing the identities \eqref{eq:identity_unit_KV_1} and \eqref{eq:identity_unit_KV_2}, the Einstein condition reduces to 
\begin{align}\label{eq:red_Einstein_cond}
    \beta^2 \mathring{R}_{ij}y^iy^j + A\left[A\left(\tfrac{1}{4}s_{ij}s^{ij} - \lambda \right) + \beta y^i\nab^ks_{ik}\right] = 0.
\end{align}
By a similar argument to the one above, it follows from the irreducibility of $A$ that $\mathring{R}_{ij} = \kappa\,a_{ij}$ for some $\kappa:M\to\R$. Hence, $a$ is Einstein with Einstein coefficient $\kappa$, which must be constant according to the Schur theorem. Next, note that \eqref{eq:red_Einstein_cond} can also be written as
\begin{align}\label{eq:red_Einstein_con_2}
    A^2\left(\tfrac{1}{4}s_{ij}s^{ij} - \lambda\right)+ \beta\left[\beta \mathring{R}_{ij}y^iy^j
    + A y^i\nab^ks_{ik}\right] = 0.
\end{align}
Since $\beta$ is also irreducible, it follows from this that $\tfrac{1}{4}s_{ij}s^{ij} - \lambda$ must be divisible by $\beta$. (Note that $A$ cannot be divisible by $\beta$ because %that would render $a$ of at most rank 2 and hence degenerate.) 
 it is irreducible.)However, since $\tfrac{1}{4}s_{ij}s^{ij} - \lambda$ is a polynomial of degree $0$ in $y$, this is only possible when it vanishes identically. It follows, using \eqref{eq:identity_unit_KV_4}, that 
\begin{align}
    \lambda = \tfrac{1}{4}s_{ij}s^{ij} = \tfrac{1}{4}\mathring{R}_{ij}b^ib^j = \tfrac{1}{4}\kappa \langle b,b\rangle = \tfrac{1}{4}\kappa.
\end{align}
Thus we have shown that $a$ is Einstein, $b$ Killing, and the Einstein constants of $L$ and $a$ are related by $\lambda =  \tfrac{1}{4}\kappa$, proving the first implication of the theorem in the case $n\geq 3$.\\

\noindent\textbf{The case $n=2$}\\
In 2D indefinite (hence Lorentzian) signature, the argument above by itself does not guarantee that $A$ is irreducible. If $A$ were irreducible, the proof given above would generalize verbatim to this case. So we may assume w.l.o.g. that $A$ is \textit{not} irreducible, i.e. $A = ql$ for two $y-$linear functions $q = q_i(x)y^i$ and $l = l_i(x)y^i$, and it suffices to derive a contradiction. Note that $q$ and $l$ \textit{are} irreducible, being linear polynomials. Hence the Einstein condition \eqref{eq:Einstein_cond_Krop} now implies that if $r_{ij}y^iy^j$ is not divisible by $A$ then $\beta$ must be divisible by $q$ or $l$. W.l.o.g we assume $\beta$ to be divisible by $q$, which says that $q_i = \gamma b_i$ for some $\gamma\in\R$. But this means that $A = \gamma \beta l$. For the sake of readability we will absorb the $\gamma$ into the definition of $l$ from here onwards, i.e. $A = \beta l$. Taking two vertical derivatives yields $2a_{ij} = b_il_j + l_ib_j$. Contracting with $b^ib^j$ and with $a^{ij}$, respectively, yields
\begin{align}
    b_il^i = 1,\qquad b_il^i = n. 
\end{align}
Together these imply that $n=1$, contradicting our assumption that $n=2$, as desired. \\

\noindent\textbf{The converse implication}\\
For the other implication, assume that $a$ is Einstein with Einstein constant $\kappa$ and $b$ is Killing. Then, employing \eqref{eq:identity_unit_KV_3} and the fact, derived above, that $s_{ij}s^{ij} = \kappa$, the LHS of the Einstein condition \eqref{eq:red_Einstein_cond} becomes
\begin{align}
    \beta^2 \kappa\,a_{ij}y^i y^j + A^2\left(\tfrac{1}{4}\kappa - \lambda \right) + A \beta y^i(-\kappa b_i) = \kappa A\beta^2   + A^2\left(\tfrac{1}{4}\kappa - \lambda \right) -\kappa A \beta^2  = A^2\left(\tfrac{1}{4}\kappa - \lambda \right),
\end{align}
which vanishes for $\lambda = \tfrac{1}{4}\kappa$. Hence, $L$ is Einstein with Einstein constant $\lambda = \tfrac{1}{4}\kappa$. This completes the proof.

\end{proof}

\section{Proof of Lemma \ref{lem:unitKV_const_curv}}\label{app:unitKV_const_curv}

Here we sketch the proof of Lemma \ref{lem:unitKV_const_curv}, repeated below for clarity:
{%keep within accolades because this numbering of the lemma is different from the rest
\renewcommand{\theprop}{\ref{lem:unitKV_const_curv}}
\begin{lem}
    Any Riemannian or Lorentzian constant curvature space admitting a unit Killing vector is either flat or locally isometric to:
\begin{itemize}
    \item an odd-dimensional round sphere $S^n$ with $n\geq 3$;
    \item an odd-dimensional $AdS_n$ with $n\geq 3$.
\end{itemize}
\end{lem}
}
\begin{proof}
     Here we fix our signature conventions as $(+\ldots +)$ (Riemannian) or $(-+\ldots+)$ (Lorentzian); together with our other conventions (see \S \ref{sec:conventions}), this ensures that the sphere and de Sitter space have positive Einstein constant and sectional curvature. The lemma follows from the following claims.
    \begin{enumerate}
        \item If $g$ is a Riemannian Einstein metric with a negative Einstein constant, then $g$ does not admit a unit KV.\\

        \textit{Proof.} Assuming we have a unit Killing vector field, the identity \eqref{eq:2nd_cov_der_of_Killing_field} leads to $b^j\nab_k\nab_jb^k = \mathring{R}_{j\ell}b^jb^\ell = \kappa \langle b,b\rangle$, so if $\kappa<0$ then the LHS must be $<0$. On the other hand, using the identities in Appendix \ref{app:identities}, the LHS can be rewritten as $b^j\nab_k\nab_jb^k = \nab_kb_j\nab^kb^j$, which is nonnegative since this is just twice the norm of the 2-form $\tfrac{1}{2}\nab_kb_j \D x^k\wedge\D x^j$. Hence $\kappa\geq 0$.
        \item If $g$ is a Lorentzian Einstein metric with positive Einstein constant, then $g$ does not admit a timelike unit KV;\\
        
        \textit{Proof.} As above, we again have $\nab_kb_j\nab^kb^j = \kappa \langle b,b\rangle$. Assuming $b$ is timelike, $\langle b,b\rangle<0$, it suffices to show that the  LHS is nonnegative. Choose an orthonormal frame $\{e_k\}_{k=0}^{n-1}$ adapted to $b$, i.e. $e_0 = b$. Then using the identities in Appendix \ref{app:identities} one finds that in this frame, $\nab_0b_0 = \nab_0b_k = \nab_kb_0 = 0$ and hence $\nab_kb_j\nab^kb^j =  \sum_{k,j=1}^{n-1}  (\nab_kb_j)^2\geq 0$.

        \item If $g$ is an even-dimensional Riemannian metric of positive constant curvature, then $g$ does not admit a unit KV;

        \textit{Proof.} Contracting the identity \eqref{eq:2nd_cov_der_of_Killing_field} with $b^j$, combining it with the identity $\mathring{R}_{ijk\ell} = K(g_{ik}g_{j\ell}-g_{i\ell}g_{jk})$ where $K$ is the constant sectional curvature, and writing it terms of the endomorphism $J$ introduced at the beginning of \S \ref{section: classification vacuum solutions} yields
        \begin{align}
            -(J^2)^k{}_i = K(\langle b,b\rangle \delta^k_i - b^k b_i).
        \end{align}
        If we regard this as an equation between linear maps and restrict it to $b^\perp$, we have $-J^2|_{b^\perp} = K \langle b,b\rangle\,\id_{b^\perp}$. Since $J^2$ maps $b$ to $0$, and $b^\perp$ to itself, $J^2$ being a square implies that $J^2|_{b^\perp}$ must also be a square, i.e. $J^2|_{b^\perp} = \br B^2 \er = -K \langle b,b\rangle\,\id_{b^\perp}$ for some endomorphism $B$ of $b^\perp$. Taking determinants, $0\leq (\det B)^2 = (-K\langle b,b\rangle)^{\dim M-1}$. Hence, on an even-dimensional Riemannian manifold, $K>0$ yields a contradiction.
        
        \item If $g$ is an even-dimensional Lorentzian metric of positive (resp. negative) constant curvature, then $g$ does not admit a spacelike (resp. timelike) unit KV. \\

        \textit{Proof.} As above, we have again the inequality $0\leq  (-K\langle b,b\rangle)^{\dim M-1}$. In this case, the combination of $\langle b,b\rangle>0$ (resp $\langle b,b\rangle<0)$, a positive (resp. negative) $K$, and even $\dim M$ leads to a contradiction.        
    
    \end{enumerate}
    Claims 1-4 leave only several possibilities. In Riemannian signature (claims 1 and 3), unless the metric is flat, the manifold must be odd-dimensional with positive curvature, hence it must be locally isometric to an odd-dimensional sphere, where $\dim M=1$ is excluded since $S^1$ is flat. In Lorentzian signature (claims 2 and 4), if the metric is not flat, the following situations are possible. First, if the Einstein constant is positive (hence $M$ is locally de Sitter), $b$ must be spacelike and $\dim M$ odd. Claim 5 below shows that this is impossible. Second, if the Einstein constant is negative (hence $M$ is locally $AdS$), then either $\dim M$ is odd (as desired), or $\dim M$ is even. Claim 6 below shows that the latter is also impossible. This leaves only the allowed cases stated in the lemma, completing the proof.

    \begin{enumerate}
        \item[5.] $dS_n$ ($n\geq 2$) does not (even locally) admit a spacelike unit Killing vector field;\\
        
        \textit{Proof.} This may be proven using embedding coordinates and writing down the explicit expression for a generic Killing vector, which will be a linear combination of rotation and boost generators, and then showing that it is not possible for such a linear combination to have constant positive norm. In the case of $dS$ (in contrast to $AdS$ below), this is a straightforward computation and we omit it here.
        
        \item[6.] $AdS_n$ ($n\geq 2$) with $n$ even does not (even locally) admit a unit norm Killing vector field.\\

        \textit{Proof.} The case $n=2$ is already excluded by the proof of Corollary \ref{cor:unitKV_2D}, so we may assume that $n\geq 3$. A Killing vector on an open subset of $AdS_n$ can be written, in embedding coordinates, as $b = (B x)^i\partial_i$ for some matrix $B$ in the Lie algebra $\mathfrak{so}(n-1,2)$, hence satisfying $B^T\eta = -\eta B$, where $\eta = \text{diag}(1,\dots,1,-1,-1)$. The norm of $b$ is given by the quadratic form $\langle b,b\rangle = (B x)^T\eta B x = x^T(B^T\eta B)x$, which is assumed to be constant on an open subset of $AdS_n$. Now, $b = (B x)^i\partial_i$ clearly extends to a Killing vector on all of $AdS_n$, and by analytic continuation, the norm $\langle b,b\rangle$ must be constant on all of $AdS_n$. %However, it follows e.g. from Hilbert's Nullstellensatz that the only quadratic forms that are constant on the defining quadric $x^T\eta x=-1$ are scalar multiples of the defining quadratic form. \\ 
        However, constancy of $\langle b,b\rangle$ on $AdS_n$ implies that $0 = \D (\langle b,b\rangle)(v) = 2 x^TB^T\eta B v$ for any $v$ in the tangent space to $AdS_n$ at $x$, i.e. whenever $x^T\eta v = 0$. In other words, $(\eta x)^\perp\subset (B^T\eta Bx)^\perp$, hence (taking another orthogonal complement), span($B^T\eta Bx$)$\subset$ span($\eta x$), so $B^T\eta Bx = \alpha(x)\eta x$ for $x\in AdS_n$. But then for any $y = rx$ with $0\neq r\in\R$, we have $B^T\eta By = rB^T\eta Bx = r\alpha(x)\eta x = \psi(y)\eta y$ with $\psi(y) = \alpha(y/r)$. We will show that $\psi$ must be constant. All timelike vectors of the ambient space are of the $rx$ with $x\in AdS_n$. Hence if $v,w$ are two linearly independent timelike vectors such that their sum is also timelike, then the requirement that $y\mapsto \psi(y)\eta y$ is linear, yields $\psi(v) = \psi(w)$. Now, any timelike vector $v$ has a timelike open neighborhood $W$ such that $v+w$ is timelike for any $w \in W$; hence $\psi$ is constant on $W$. But then the linear map $B^T\eta B$ agrees with a scalar multiple of the identity on the open set $W$; by linearity, it must therefore equal that scalar multiple of the identity on the whole ambient space. Hence $B^T\eta B = \psi \eta$ and then it follows that, with $x\in AdS$, $\langle b,b\rangle = x^TB^T\eta Bx = \psi x^T \eta x = -\psi$. Taking determinants, we obtain $\det B^2 = (-\langle b,b\rangle)^{n+1}$. On the other hand, $B^T\eta = -\eta B$ implies that $\det B = (-1)^{n+1}\det B$, which for even $n$ implies that $\det B = 0$, hence $\langle b,b\rangle = 0$.
    \end{enumerate}
    
\end{proof}

\bibliographystyle{utphys}
\bibliography{bib}

@article{ZHANG201380,
title = "{On Einstein-Kropina metrics}",
journal = {Differential Geometry and its Applications},
volume = {31},
number = {1},
pages = {80-92},
year = {2013},
issn = {0926-2245},
doi = {https://doi.org/10.1016/j.difgeo.2012.10.011},
url = {https://www.sciencedirect.com/science/article/pii/S0926224512000988},
author = {Xiaoling Zhang and Yi-Bing Shen}
}

@Book{ChernShen2005,
  author    = {Shiing-Shen Chern and Zhongmin Shen},
  title     = {Riemann-Finsler geometry},
  publisher = {World Scientific},
  year      = {May 2005},
  series    = {Nankai Tracts in Mathematics Vol. 6},
}

@article{villasenor:2025,
author = "F. Villase\~nor, Fidel",
title = "{Schur theorem for the Ricci curvature of any weakly Landsberg Finsler metric}",
journal = {Israel Journal of Mathematics},
pages = {1-25},
year = {2025},
doi = {https://doi.org/10.1007/
s11856-025-2797-z},

}

@phdthesis{villasenor:2024,
  title        = "{Lorentz-Finsler geometry and Einstein equations}",
  author       = "F. Villase\~nor, Fidel",
  year         = "April 2024",
  address      = "Spain",
  note         = {Available at \url{https://digibug.ugr.es/handle/10481/94842}},
  school       = "University of Granada",
  type         = {PhD thesis}
}

@inproceedings{BaoRobles,
  title        = "{Ricci and flag curvatures in Finsler geometry}",
  author       = {Bao, David and Robles, Colleen},
  year         = 2004,
  booktitle    = {A sampler of Riemann-Finsler geometry},
  publisher    = {Cambridge University Press},
  series       = {Mathematical Sciences Research Institute
Publications},
  volume       = 50,
  pages        = {197--259},
  editor       = "{D. Bao et al.}",
  doi          = {10.1017/9781009701280.006}
}

@article{Pfeiferetal2025,
  author    = {Pfeifer, Christian and Voicu, Nicoleta and Friedl-Sz{\'a}sz, Annam{\'a}ria and Popovici-Popescu, Elena},
  title     = {From kinetic gases to an exponentially expanding universe: the Finsler-Friedmann equation},
  journal   = {Journal of Cosmology and Astroparticle Physics},
  year      = {2025},
  volume    = {2025},
  number    = {10},
  pages     = {050},
  doi       = {10.1088/1475-7516/2025/10/050},
  publisher = {IOP Publishing}
}

@article{Voicu_spacetime_cond_2023,
AUTHOR = {Voicu, Nicoleta and Friedl-Sz{\'a}sz, Annam{\'a}ria and Popovici-Popescu, Elena and Pfeifer, Christian},
TITLE = {The {F}insler Spacetime Condition for $(\alpha,\beta)$-Metrics and Their Isometries},
JOURNAL = {Universe},
VOLUME = {9},
YEAR = {2023},
NUMBER = {4},
ARTICLE-NUMBER = {198},
URL = {https://www.mdpi.com/2218-1997/9/4/198},
ISSN = {2218-1997},
DOI = {10.3390/universe9040198}
}

@article{Cheng2017,
  author    = {Cheng, XinYue and Shen, ZhongMin},
  title     = {{Einstein Finsler metrics and Killing vector fields on Riemannian manifolds}},
  journal   = {Science China Mathematics},
  year      = {2017},
  volume    = {60},
  number    = {1},
  pages     = {83--98},
  doi       = {10.1007/s11425-016-0303-6},
  url       = {https://doi.org/10.1007/s11425-016-0303-6},
  issn      = {1869-1862}
}

@article{Hawking:1976jb,
    author = "Hawking, S. W.",
    title = "{Gravitational instantons}",
    reportNumber = "Print-77-0294 (CAMBRIDGE)",
    doi = "10.1016/0375-9601(77)90386-3",
    journal = "Phys. Lett. A",
    volume = "60",
    pages = "81",
    year = "1977"
}

@article{Einstein_warped_Shen,
author = {Chen, Bin and Shen, Zhongmin and Zhao, Lili},
title = {Constructions of {Einstein Finsler} metrics by warped product},
journal = {International Journal of Mathematics},
volume = {29},
number = {11},
pages = {1850081},
year = {2018},
doi = {10.1142/S0129167X18500817}
}

@article{TANG201883,
title = {Some remarks on {Einstein-Randers} metrics},
journal = {Differential Geometry and its Applications},
volume = {58},
pages = {83-102},
year = {2018},
issn = {0926-2245},
doi = {https://doi.org/10.1016/j.difgeo.2018.01.002},
url = {https://www.sciencedirect.com/science/article/pii/S0926224518300329},
author = {Xiaoyun Tang and Changtao Yu}
}

@article{Sasaki-Einstein,
    author = {James Sparks},
    title = "{Sasaki-Einstein manifolds}",
    journal = "Surveys in Differential Geometry",
    volume = {16},
    pages = {265-324},
    year = {2011},
    doi = {https://doi.org/10.48550/arXiv.1004.2461},
    url = {https://arxiv.org/abs/1004.2461}
}

@incollection{Sasaki-Einstein-AdS/CFT,
  author    = {Gauntlett, Jerome P. and Martelli, Dario and Sparks, James},
  title     = {{S}asaki-{E}instein Geometry, {GK} Geometry and the {AdS/CFT} Correspondence},
  booktitle = {Half a Century of Supergravity},
  editor    = {Ceresole, A. and Dall'Agata, G.},
  year      = {2026},
  publisher = {Cambridge University Press},
  note      = {Invited contribution},
  eprint    = {2503.01950},
  archivePrefix = {arXiv},
  primaryClass  = {hep-th},
  doi       = {10.48550/arXiv.2503.01950}
}

@article{Sasaki-Einstein_on_S3xS2,
    author = "Gauntlett, Jerome P. and Martelli, Dario and Sparks, James and Waldram, Daniel",
    title = "{Sasaki-Einstein metrics on {$S^2 \times S^3$}}",
    eprint = "hep-th/0403002",
    archivePrefix = "arXiv",
    doi = "10.4310/ATMP.2004.v8.n4.a3",
    journal = "Adv. Theor. Math. Phys.",
    volume = "8",
    number = "4",
    pages = "711--734",
    year = "2004"
}

@article{Gauntlett:2004hh,
    author = "Gauntlett, Jerome P. and Martelli, Dario and Sparks, James F. and Waldram, Daniel",
    title = "{A new infinite class of {S}asaki-{E}instein manifolds}",
    eprint = "hep-th/0403038",
    archivePrefix = "arXiv",
    doi = "10.4310/ATMP.2004.v8.n6.a3",
    journal = "Adv. Theor. Math. Phys.",
    volume = "8",
    number = "6",
    pages = "987--1000",
    year = "2004"
}

@Article{axioms8010006,
AUTHOR = {Perrone, Domenico},
TITLE = {Contact Semi-{R}iemannian Structures in {CR} Geometry: Some Aspects},
JOURNAL = {Axioms},
VOLUME = {8},
YEAR = {2019},
NUMBER = {1},
ARTICLE-NUMBER = {6},
URL = {https://www.mdpi.com/2075-1680/8/1/6},
DOI = {10.3390/axioms8010006}
}

@Article{Lambda-vacuum-paper,
AUTHOR = {F. Villase\~nor, Fidel and Heefer, Sjors and Fuster, Andrea},
TITLE = {Finsler gravity with a cosmological `constant'},
note    = {In preparation},
}

@misc{MartinGarcia:xAct,
  author       = {Mart\'in-Garc\'ia, Jos\'e M.},
  title        = {{xAct:} Efficient tensor computer algebra for {Mathematica}},
  howpublished = {\url{http://xact.es/}},
  year         = {2002--2018},
  note         = {Software package}
}

@phdthesis{Heefer:2024Thesis,
    author = "Heefer, Sjors",
    title = "{Finsler geometry, spacetime \& Gravity -- From metrizability of Berwald spaces to exact vacuum solutions in Finsler gravity}",
    eprint = "2404.09858",
    School = "Eindhoven University of Technology",
    archivePrefix = "arXiv",
    primaryClass = "gr-qc",
    year = "2024"
}

@article{Heefer_2023_Finsler_grav_waves,
doi = {10.1088/1361-6382/acecce},
year = {2023},
publisher = {IOP Publishing},
volume = {40},
_number = {18},
pages = {184002},
author = {S. Heefer and Andrea Fuster},
title = {Finsler gravitational waves of $(\alpha,\beta)$-type and their observational signature},
journal = {Class. Quantum Gravity},
	eprint         = "2302.08334",
	archivePrefix  = "arXiv"
}

@article{Pfeifer:2011xi,
	author         = {Pfeifer, C. and Wohlfarth, M. N. R.},
	title          = "{Finsler geometric extension of Einstein gravity}",
	journal        = "Phys. Rev. D",
	volume         = "85",
	pages          = "064009",
	year           = "2012",
	eprint         = "1112.5641",
	archivePrefix  = "arXiv",
	_primaryClass   = "gr-qc",
    doi = "https://doi.org/10.1103/PhysRevD.85.064009"
}

@article{Hohmann_2019,
  title = {Finsler gravity action from variational completion},
  author = {Hohmann, Manuel and Pfeifer, Christian and Voicu, Nicoleta},
  journal = {Phys. Rev. D},
  volume = {100},
  _issue = {6},
  pages = {064035},
  numpages = {18},
  year = {2019},
  _month = {Sep},
  publisher = {American Physical Society},
  doi = {10.1103/PhysRevD.100.064035},
	eprint         = "1812.11161",
	archivePrefix  = "arXiv"
}

@phdthesis{Robles_2003, 
series={Retrospective Theses and Dissertations, 1919-2007}, 
title="{Einstein metrics of Randers type}", 
url={https://open.library.ubc.ca/collections/ubctheses/831/items/1.0079500}, DOI={http://dx.doi.org/10.14288/1.0079500}, 
school={University of British Columbia}, 
author={Robles, Colleen}, 
year={2003}, 
collection={Retrospective Theses and Dissertations, 1919-2007}}

@article{YU2010290,
title = {On {E}instein $m$-th root metrics},
journal = {DiffEerential Geometry and its Applications},
volume = {28},
number = {3},
pages = {290-294},
year = {2010},
issn = {0926-2245},
doi = {https://doi.org/10.1016/j.difgeo.2009.10.011},
url = {https://www.sciencedirect.com/science/article/pii/S092622450900120X},
author = {Yaoyong Yu and Ying You}
}

@article{ZhangXia2014,
  author  = {Zhang, X. and Xia, Q.},
  title   = "{On Einstein Matsumoto metrics}",
  journal = {Science China Mathematics},
  volume  = {57},
  pages   = {1517--1524},
  year    = {2014},
  doi     = {10.1007/s11425-014-4788-0}
}

@book{besse1987einstein,
  author    = {Arthur L. Besse},
  title     = {Einstein manifolds},
  series    = {Ergebnisse der Mathematik und ihrer Grenzgebiete. 3. Folge},
  volume    = {10},
  publisher = {Springer-Verlag},
  year      = {1987},
  doi       = {10.1007/978-3-540-74311-8}
}

@incollection{Bao2007,
  author    = {David Bao},
  title     = {On two curvature-driven problems in {R}iemann-{F}insler geometry},
  booktitle = "{Finsler geometry, Sapporo 2005 --- In memory of Makoto Matsumoto}",
  series    = {Advanced Studies in Pure Mathematics},
  volume    = {48},
  pages     = {19--71},
  year      = {2007},
  publisher = {Mathematical Society of Japan},
  doi       = {10.2969/aspm/04810019}
}

@article{ZShen2014,
  author    = {Shen, Zhongmin and Yang, Guojun},
  title     = {On a class of weakly {E}instein {F}insler metrics},
  journal   = {Israel Journal of Mathematics},
  year      = {2014},
  volume    = {199},
  pages     = {773--790},
  doi       = {10.1007/s11856-013-0060-5},
  publisher = {Springer}
}

\end{document}